\documentclass[10pt,twocolumn,twoside,english]{IEEEtran}
\usepackage[T1]{fontenc}
\usepackage{mathrsfs}
\usepackage{bm}
\usepackage{amsmath}
\usepackage{amsthm}
\usepackage{amssymb}
\usepackage{graphicx}
 \PassOptionsToPackage{normalem}{ulem}
\usepackage{hyperref}
\usepackage{bbm}
\makeatletter

\theoremstyle{plain}
\newtheorem{theorem}{\protect\theoremname}
  \theoremstyle{plain}
  
  \theoremstyle{plain}
  
  \theoremstyle{plain}
   \newtheorem{lemma}{\protect\lemmaname}
  \theoremstyle{remark}
  \newtheorem{remark}{\protect\remarkname}
   
\theoremstyle{assumption}
    \theoremstyle{proposition}

\theoremstyle{algorithm}  
  
\usepackage{amsmath}
\usepackage{amssymb}
\usepackage{graphicx,psfrag,cite,subfigure}
\usepackage[table]{xcolor}
\setkeys{Gin}{width=1.0\columnwidth}

\makeatother
  \providecommand{\definitionname}{Definition}
  \providecommand{\lemmaname}{Lemma}
  \providecommand{\propositionname}{Proposition}
  \providecommand{\remarkname}{Remark}
\providecommand{\theoremname}{Theorem}
\providecommand{\conjecturename}{Conjecture}
\providecommand{\assumptionname}{Assumption}
\providecommand{\algorithmname}{Algorithm}

\begin{document}

 \title{Quickest Bayesian and non-Bayesian detection of false data injection attack in remote state estimation
  \thanks{Akanshu Gupta is currently with Microsoft India. Work done while at Indian Institute of Technology, Delhi. Email: akanshu3299@gmail.com }
  \thanks{Abhinava Sikdar is with the Department of Computer Science, Columbia University. Email: as6413@columbia.edu}
 \thanks{Arpan Chattopadhyay is with the Department of Electrical Engineering and the Bharti School of Telecom Technology and Management, IIT Delhi. Email: arpanc@ee.iitd.ac.in }
 \thanks{This work was supported by the  faculty seed grant and  professional development allowance  of Arpan Chattopadhyay at IIT Delhi. }
 \thanks{The conference precursor \cite{gupta2020quickest} of this work was accepted in IEEE International Symposium on Information Theory (ISIT), 2021.}
}

\author{
 Akanshu Gupta \hspace{0.2cm} Abhinava Sikdar  \hspace{0.2cm}Arpan Chattopadhyay \vspace*{-0.3in}
}

\maketitle
%
%



\ifdefined\SINGLECOLUMN
	\setkeys{Gin}{width=0.5\columnwidth}
	\newcommand{\figfontsize}{\footnotesize} 
\else
	\setkeys{Gin}{width=1.0\columnwidth}
	\newcommand{\figfontsize}{\normalsize} 
\fi

\begin{abstract} 
In this paper, quickest detection of false data injection attack on remote state estimation is considered. A set of $N$ sensors make noisy linear observations of a discrete-time linear process with Gaussian noise, and report the observations to a remote estimator. The challenge is the presence of a few potentially malicious sensors which can start strategically manipulating their observations at a random time  in order to skew the estimates.  The quickest attack detection problem  for a  known {\em linear} attack scheme in the Bayesian setting with a Geometric prior on the attack initiation instant is posed as a constrained Markov decision process (MDP),  in order to minimize the expected detection delay subject to a false alarm constraint, with the state involving the  probability belief at the estimator that the system is under attack. State transition probabilities are derived in terms of system parameters, and the structure of the  optimal policy is derived  analytically. It turns out that the optimal policy amounts to checking whether the  probability belief  exceeds a threshold. Next, generalized CUSUM based attack detection algorithm is proposed for the non-Bayesian setting where the attacker chooses the attack initiation instant in a particularly  adversarial manner. It turns out that computing the statistic for the generalised CUSUM test in this setting  relies on the same techniques developed to compute the state transition probabilities of the MDP. Numerical results demonstrate significant performance gain under the proposed algorithms  against competing algorithms.
\end{abstract}

\section{Introduction}\label{section:introduction}
Networked estimation and control of physical processes and systems are indispensable components of cyber-physical systems (CPS) that involve integration of sensing,  computation,  communication and control to realize the ultimate combining of the physical systems and the cyber world. The applications of CPS are many-fold: intelligent transportation systems,  smart grids, networked  monitoring and control of industrial processes, environmental monitoring,  disaster management, etc. These applications heavily depend on reliable    estimation of a physical process or system  via   sensor observations collected over a wireless network \cite{chattopadhyay2020dynamic}.  However,    malicious attacks on these sensors pose a major security threat to CPS. One such common attack is  a {\em denial-of-service (DoS)} attack where the attacker attempts to block system resources  ({\em e.g.}, wireless jamming attack  \cite{guan2018distributed}). Contrary to DoS, we focus on {\em false data injection} (FDI) attacks which is a specific class of  integrity or deception attacks, where the sensor observations are modified before they are sent to the remote estimator  \cite{mo2009secure, mo2014detecting}.  The attacker can modify  the information   by breaking the cryptography of the   packets or by physically manipulating   the sensors ({\em e.g.}, placing  a cooler near a temperature sensor).

\begin{figure}[t!]
\begin{centering}
\begin{center}
\includegraphics[height=4.5cm, width=\linewidth]{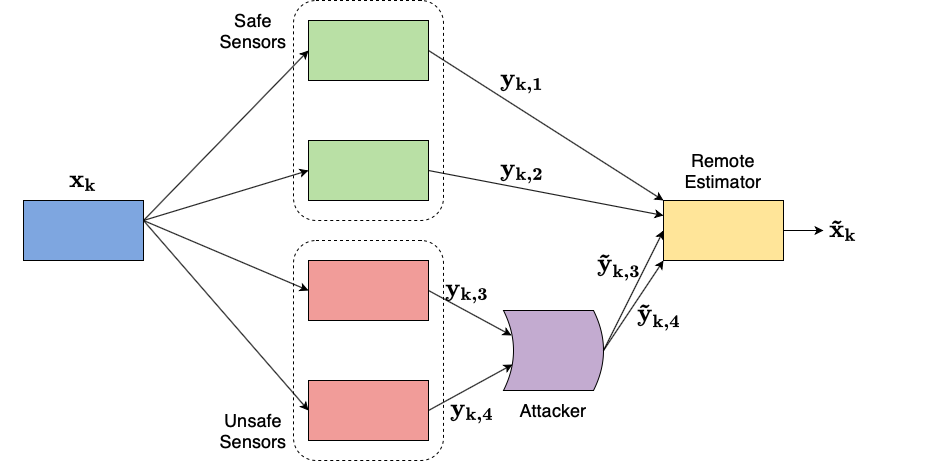}
\end{center}
\end{centering}
\vspace{-5mm}
\caption{False data injection attack in centralized  remote estimation.}
\label{fig:FDI-attack-image}
\vspace{-5mm}
\end{figure}

Recently, the problem of FDI attack and its countermeasures has received significant attention \cite{ding2020secure}. The literature in this area can be classified into two categories: (i) FDI on centralized systems, and (ii) FDI on distributed systems. In a distributed system, various entities  such as sensors, estimators, controllers and actuators are connected via a multi-hop wireless network, and FDI on one such component propagates to other components over time via the network. 

There has been a vast literature on FDI in centralized systems, particularly in the remote estimation setting. Such works include   several attempts to characterize and design   attack schemes:  undetectable linear deception attack \cite{guo2017optimal} in single sensor context, and also  conditions for    undetectable FDI  attack  \cite{chen2017optimal}. The     attack strategy to steer the control of CPS to a desired value under attack detection constraint  is provided in \cite{chen2016cyber}.  Literature on attack detection considered a number of   models and approaches; {\em e.g.}, attack    detection schemes for {\em noiseless} systems  \cite{pasqualetti2013attack}, comparing the observations from a few {\em known safe} sensors against potentially malicious sensor observations  \cite{li2017detection} to tackle the attack of \cite{guo2017optimal}, coding of   output of sensors  along with $\chi^2$ detector   \cite{miao2017coding}, Gaussian mixture model based detection and secure state estimation \cite{guo2018secure},  and the innovation vector based attack detection and secure estimation schemes \cite{mishra2017secure}. There have also been a number of works on secure estimation: see    \cite{pajic2017attack} for bounded noise case, \cite{liu2017dynamic} for  sparsity models to characterize the   location switching attack in a {\em noiseless}   system and    state recovery constraints, \cite{nakahira2018attack} for noiseless systems, \cite{chattopadhyay2019security, chattopadhyay2018secure, chattopadhyay2018attack} for linear Gaussian process and linear observation with Gaussian noise.   Attack detection,   estimation and control  for power  systems are addressed in \cite{manandhar2014detection, liang2017review, hu2017secure}.  Attack-resilient  control under FDI for noiseless systems is discussed in \cite{fawzi2014secure}.

Though relatively new, FDI on distributed systems is also increasingly being investigated;  see  \cite{guan2017distributed} for attack detection and secure estimation,  \cite{satchidanandan2016dynamic} for attack detection in networked control system via  {\em dynamic watermarking},   \cite{ge2019distributed} for distributed Krein space based attack detection in discrete time-varying systems,  and \cite{dorfler2011distributed} for distribured attack detection in power systems. On the other hand,   attack design for  distributed CPS is also being investigated; see \cite{choraria2019optimal} for linear  attack design against distributed state estimation to push all nodes' estimates to a desired target under a given attack detection constraint,    \cite{moradi2019coordinated} for attack design to maximize the network-wide estimation error via simple Gaussian noise addition, \cite{lu2019malicious} for conditions for perfect attack in a distributed  control system and  design algorithms for perfect and non-perfect attacks, etc.

While there have been several schemes (such as \cite{li2017detection} and the $\chi^2$ detector) to detect FDI attack in a multi-sensor setting, the optimal attack detector is not yet known even for the centralized systems. Moreover, the popular $\chi^2$ detector fails if the attacker judiciously injects false data in such a way that the innovation sequence of the Kalman state  estimator remains constant over time.   In light of this issue, our main contributions in this paper are as follows:
\begin{enumerate}
    \item In a multi-sensor setting with some safe and some potentially unsafe sensors, for the linear attack of \cite{guo2017optimal} with known attack parameters, we develop an optimal Bayesian  attack detector that minimizes the mean delay in attack detection subject to a constraint on the false alarm probability.   The problem is formulated as a partially observable Markov decision process (POMDP), and the optimal policy turns out to be a simple threshold policy on the belief probability that   an attack has already been launched.  
    
\item Though POMDP based Bayesian change detection techniques exist in the literature \cite{poor2009quickest}, our problem involves a far more complicated state that consists of the belief and the collection of past Kalman innovations; this occurs primarily because of the temporal dependence of innovation sequence after the attack and the uncertainty in the attack initiation instant. Also, we do not apply Shiryaev's test directly because that would be optimal only for $\alpha \rightarrow 0$ and not necessarily for all $\alpha$, where $\alpha$ is the probability of false alarm. 
    \item Computing the  belief probability recursively is a challenging problem due to the complicated temporal dependence of the innovation sequence available to the estimator, and we solve this by modeling the post-attack state estimation process as a  Kalman filter with a modified process and observation model. Our results in Section~\ref{section:recovering} explicitly show how to compute the true post-attack innovation distribution. 
    \item In order to reduce computational complexity, we propose a sub-optimal algorithm called QUICKDET  whose structure is motivated by the optimal policy derived from the POMDP formulation. QUICKDET applies a constant threshold rule on the belief, while the POMDP formulation yields a time-varying threshold. We provide a simulation-based technique to optimize this constant threshold, by using tools from simultaneous perturbation stochastic approximation (SPSA \cite{spall92original-SPSA}) and two timescale stochastic approximation \cite{borkar08stochastic-approximation-book}. Numerical results show that this sub-optimal algorithm significantly outperforms competing algorithms.
    \item The Bayesian quickest FDI detection algorithm is also extended to the case where the attacker can launch an  attack using multiple possible linear attack strategies.
    \item In the non-Bayesian setting, we adapt the generalized CUSUM test from the literature for quickest detection of FDI. It turns out that the test statistic for generalized CUSUM in our problem can be computed by using the same tools developed in Section~\ref{section:recovering}. 
\end{enumerate}
The rest of the paper is organized as follows. The system model is described in Section~\ref{section:system-model}. Section~\ref{section:recovering} develops the necessary theory for calculating the state transition probabilities for the POMDP formulation. POMDP formulation for FDI attack detection in the Bayesian setting with known linear attack parameters is provided in Section~\ref{section:known-attack-scheme}. Quickest attack detection algorithm for the non-Bayesian setting is provided in Section~\ref{section:quickdet-non-Bayesian}.   Numerical results are provided in Section~\ref{section:numerical-results}, followed by the conclusions in Section~\ref{section:conclusion}. All proofs are provided in the appendices.

\section{System model}\label{section:system-model}
In this paper, bold capital letters, bold small letters,  and capital letters with caligraphic font   will denote matrices, vectors and sets respectively. For any two square matrices $\bm{M}_1$ and $\bm{M}_2$, the block diagonalization operator is defined by  
$Blkdiag(\bm{M}_1, \bm{M}_2) \doteq \begin{bmatrix}
           \bm{M}_1 \,\,\,\,  \bm{0} \\
           \,\, \bm{0}  \,\,\,\,\, \bm{M}_2
         \end{bmatrix}$. Also, the transpose of a matrix $\bm{M}$ is denoted by $\bm{M}'$.

\subsection{Sensing and estimation model}
We consider a set of sensors $\mathcal{N} \doteq \{1,2,\cdots,N\}$ sensing a discrete-time process $\{\bm{x}_k\}_{k \geq 1}$ which is a linear Gaussian process with the following dynamics:
\begin{equation}\label{eqn:process-equation}
    \bm{x}_{k+1}=\bm{A} \bm{x}_k+\bm{w}_k
\end{equation}
where $\bm{x}_k \in \mathbb{R}^{q \times 1}$ is vector-valued, $\bm{A} \in \mathbb{R}^{q \times q}$ is the process matrix, and $\bm{w}_k \sim \mathcal{N}(\bm{0}, \bm{Q})$ is the Gaussian process noise i.i.d. across $k$. 
The observation made by sensor~$i$ at time~$k$ is:
\begin{equation}\label{eqn:observation-equation}
    \bm{y}_{k,i}=\bm{C}_i \bm{x}_k + \bm{v}_{k,i}
\end{equation}
where $\bm{C}_i$ is a matrix of appropriate dimensions and $\bm{v}_{k,i} \sim \mathcal{N}(\bm{0}, \bm{R}_i)$ is the Gaussian observation noise at sensor~$i$ at time~$k$, which is i.i.d. across $k$ and independent across $i$. We assume that $(\bm{A}, \bm{Q}^{\frac{1}{2}})$ is stabilizable and $(\bm{A}, \bm{C}_i)$ is detectable for all $i \in \mathcal{N}$. 
The complete observation from all sensors at time~$k$ is denoted  by $\bm{y}_k \doteq (\bm{y}_{k,1}', \bm{y}_{k,2}', \cdots, \bm{y}_{k,N}' )'$ which can be written as:
\begin{equation}\label{eqn:consolidated-observation-model}
    \bm{y}_k =\bm{C} \bm{x}_k + \bm{v}_k
\end{equation}
where $\bm{C} \doteq  (\bm{C}_1', \bm{C}_2',\cdots,  \bm{C}_N')'$ is the equivalent observation  matrix and $\bm{v}_k$ is the zero mean Gaussian observation  noise with covariance matrix $Blkdiag(\bm{R}_1, \bm{R}_2, \cdots, \bm{R}_N)$.

\subsection{Process estimation under no attack} 

For the model in \eqref{eqn:consolidated-observation-model}, a standard Kalman filter  \cite{anderson1979optimal} is used to estimate $\hat{\bm{x}}_k$ from $\{\bm{y}_j\}_{j \leq k}$ in order to minimize the mean squared error (MSE):
\begin{eqnarray}
\bm{\hat{x}}_{k+1|k}&=& \bm{A} \bm{\hat{x}}_k \nonumber\\
\bm{P}_{k+1|k} &=& \bm{A} \bm{P}_k \bm{A}'+\bm{Q} \nonumber\\
\bm{K}_{k+1} &=& \bm{P}_{k+1 | k} \bm{C}' (\bm{C} \bm{P}_{k+1|k} \bm{C}' +\bm{R})^{-1} \nonumber\\
\bm{\hat{x}}_{k+1} &=& \bm{\hat{x}}_{k+1|k} +\bm{K}_{k+1} (\bm{y}_{k+1}-\bm{C} \bm{\hat{x}}_{k+1|k} )  \nonumber\\
\bm{P}_{k+1} &=& (\bm{I}-\bm{K}_{k+1} \bm{C}) \bm{P}_{k+1|k}, \label{eqn:kalman-filter-for-single-sensor}
\end{eqnarray}
where $\hat{\bm{x}}_{k+1}=\mathbb{E}(\bm{x}_{k+1} | \bm{y}_0,\bm{y}_1,\cdots, \bm{y}_{k+1})$ is the MMSE estimate and $\bm{P}_{k+1}$ is the the  error covariance matrix for this estimate. From \cite{anderson1979optimal}, we know that   $\lim_{k \rightarrow \infty} \bm{P}_{k+1|k} =\bm{\underbar{P}}$ exists and is the unique fixed point to the  {\em Riccati equation} which is basically the $\bm{P}_{k+1|k}$ iteration. 
Let  us also define the innovation vector from sensor~$i$ at time~$k$ as  $\bm{z}_{k,i}\doteq \bm{y}_{k,i}-\bm{C}_i \bm{\hat{x}}_{k|k-1}$, and the collective innovation as $\bm{z}_k \doteq \bm{y}_k-\bm{C} \bm{\hat{x}}_{k|k-1}$. It is well-known  \cite{anderson1979optimal} that $\{\bm{z}_k\}_{k \geq 1}$ is a zero-mean Gaussian sequence   independent across time and whose steady-state covariance matrix   is $\bm{\Sigma}_{\bm{z}}\doteq(\bm{C} \bm{\underbar{P}}\bm{C}'+\bm{R})$. Hence, in most cases,  the residue-based $\chi^2$ detector is used at the remote estimator's side for any attack or anomaly detection, which has the following mathematical form:
\begin{equation}
    \sum_{k=\tau-J+1}^{\tau} \bm{z}_k' \bm{\Sigma}_{\bm{z}}^{-1} \bm{z}_k \mathop{\gtrless}_{H_0}^{H_1} \eta
\end{equation}
where the null Hypothesis $H_0$ represents no attack or anomaly, and  $H_1$ represents the presence of attack or anomaly. Here $J$ is a pre-specified window size  and $\eta$ is a pre-specified threshold which can be  tuned to control the false alarm probability.

\subsection{FDI  attack}\label{subsection:attack-model}
If a sensor~$i$ is under FDI attack, the observation sent to the remote estimator from this sensor becomes: 
\begin{equation}\label{eqn:observation-equation-under-attack}
    \tilde{\bm{y}}_{k,i}=\bm{C}_i \bm{x}_k + \bm{v}_{k,i}+\bm{e}_{k,i}
\end{equation}
where $\bm{e}_{k,i}$ is the  false data injected by sensor~$i$ at time~$k$. Consequently,  \eqref{eqn:consolidated-observation-model} is modified to:
\begin{equation}\label{eqn:consolidated-observation-under-attack}
    \tilde{\bm{y}}_k =\bm{C} \bm{x}_k + \bm{v}_k+ \bm{e}_k
\end{equation}

Let us recall  {\em linear attacks} for the single sensor case \cite{guo2017optimal} where, at time $k$, the malicious sensor modifies the innovation  as $\bm{\tilde{z}}_k=\bm{T} \bm{z}_k+\bm{b}_k$, where $\bm{T}$ is a square matrix and $\bm{b}_k \sim N (\bm{0},\bm{\Sigma}_{\bm{b}})$ is  i.i.d. Gaussian random vector sequence. The authors of  \cite{guo2017optimal} had shown that $\bm{\tilde{z}}_k \sim N(\bm{0},\bm{\Sigma}_{\bm{\tilde{z}}})$ under steady state, where $\bm{\Sigma}_{\bm{\tilde{z}}}=\bm{T} \bm{\Sigma}_{\bm{z}} \bm{T}'+\bm{\Sigma}_{\bm{b}}$. Hence, by ensuring  $\bm{\Sigma}_{\bm{\tilde{z}}}=\bm{\Sigma}_{\bm{z}}$, one can preserve the distribution of  $\{\bm{\tilde{z}}_k\}_{k \geq 1}$ in the single sensor case, and hence the detection probability will remain same even under FDI, under the $\chi^2$ detector. It was also shown in \cite{guo2017optimal}   that inverting the sign of the innovation   ({\em i.e.},   $\bm{T}=-\bm{I}$ and $\bm{b}_k=\bm{0}$) maximizes the MSE. Obviously, $\bm{T}=\bm{I}$ and $\bm{b}_k=\bm{0}$ imply that there is no attack.
  
 In this paper, we assume that there is a set of sensors $\mathcal{S} \subset \mathcal{N}$ which are safe, i.e., the sensor belonging to $\mathcal{S}$ can not be attacked. Let the set of potentially unsafe sensors be denoted by $\mathcal{A} \doteq \mathcal{N}-\mathcal{S}$. Clearly, a general  $\bm{T}$ matrix as the single sensor case will not be useful here; instead, we assume that $\bm{T} = Blkdiag (\bm{T}_{\mathcal{S}}, \bm{T}_{\mathcal{A}})=Blkdiag ( Blkdiag(\{\bm{T}_i\}_{i \in \mathcal{S}}), Blkdiag(\{\bm{T}_i\}_{i \in \mathcal{A}})  )$ where $\bm{T}_i$ is an identity matrix of appropriate dimension if   $i \in \mathcal{S}$. Also, for $i \in \mathcal{S}$, the added noise $\bm{b}_{k,i}=\bm{0}$ for all time $k \geq 1$, and, for $i \in \mathcal{A}$, we have $\bm{b}_{k,i} \sim N(\bm{0}, \Sigma_{\bm{b}_i})$ i.i.d. across $k$. Consequently, the modified innovation at node $i \in \mathcal{A}$ at time~$k$ is given by:
 \begin{equation}\label{eqn:modified-innovation}
     \tilde{\bm{z}}_{k,i}=\bm{T}_i \bm{z}_{k,i}+ \bm{b}_{k,i}
 \end{equation}

  We  define $\bm{z}_{k,\mathcal{A}}, \tilde{\bm{z}}_{k,\mathcal{A}}, \bm{y}_{k,\mathcal{A}}, \tilde{\bm{y}}_{k,\mathcal{A}}, \bm{C}_{\mathcal{A}}$ the components of $\bm{z}_k, \tilde{\bm{z}}_k, \bm{y}_k, \tilde{\bm{y}}_k, \bm{C}$ coming from sensors of $\mathcal{A}$; similar notation is used for the components corresponding to the safe sensors $\mathcal{S}$ also. We also define by $\tilde{\bm{x}}_k$ the estimate at time~$k$ generated by a Kalman filter under FDI  attack. The history available to the remote estimator at time~$k$ is denoted by $\mathcal{H}_k \doteq \{\tilde{\bm{z}}_j: j \leq k\}$.
  
 We also define $\tilde{\bm{x}}_k $ as the   estimate of $\bm{x}_k$ by applying a standard Kalman filter on  $\mathcal{H}_k$. Also, let  $\tilde{\bm{x}}_{k| \mathrm{a}} \doteq \mathbb{E}(\bm{x}_k | \mathcal{H}_k, \mathrm{a})$ denote the MMSE estimate of $\bm{x}_k$ given $\mathcal{H}_k$ and the event $\mathrm{a}$ that attack has already happened.

  For the multi-sensor setting,  \cite{li2017detection} proposed an   attack detection algorithm under the presence of a few known {\em safe} sensors. Unfortunately, the detection algorithm of \cite{li2017detection} and other detection algorithms from the literature do not have any optimality proof. On the other hand, if the attacker ensures a constant $\tilde{\bm{z}}_k$ for all $k$ such that $\sum_{k=\tau-J+1}^{\tau} \tilde{\bm{z}}_k' \bm{\Sigma}_{\tilde{\bm{z}}}^{-1} \tilde{\bm{z}}_k < \eta$, then the $\chi^2$ detector fails. These inadequacies in the literature   motivate the quickest attack detection problem formulation in this paper.

\section{Recovering true innovation and estimates from FDI attacked observations}\label{section:recovering}
In this section, we will develop a technique for computing  $\tilde{\bm{x}}_{k| \mathrm{a}}$ under the assumption that the attack initiation instant $t$ is known;  this will also yield the  statistics of $\tilde{\bm{z}}_k$. These results are necessary to calculate the state transition probabilities of the POMDP formulation in Section~\ref{section:known-attack-scheme}. These results also show that linear attack changes the distribution of innovations and estimates in a multi-sensor scenario with a few safe sensors. Throughout this section, we assume that $\bm{T}$ is known to the remote estimator.

\subsection{Computing $\tilde{\bm{x}}_{k|\mathrm{a}}$}
The  innovation from the unsafe sensors for $k \geq t$ is: 
\begin{eqnarray}\label{eqn:equivalent-observation-unsafe-sensors}
    \tilde{\bm{z}}_{k,\mathcal{A}} &=& \bm{T}_{\mathcal{A}} (\bm{y}_{k,\mathcal{A}} - \bm{C}_\mathcal{A} \bm{A}\tilde{\bm{x}}_{k-1}) + \bm{b}_k \nonumber\\
 \implies   \bm{T}_\mathcal{A}^{-1} \tilde{\bm{z}}_{k,\mathcal{A}} 
    &=& (\bm{y}_{k,\mathcal{A}} + \bm{T}_\mathcal{A}^{-1} \bm{b}_k) - \bm{C}_\mathcal{A} \bm{A}\tilde{\bm{x}}_{k-1} \nonumber \\
    &=& \tilde{\tilde{\bm{y}}}_{k,\mathcal{A}} - \bm{C}_\mathcal{A} \bm{A}\tilde{\bm{x}}_{k-1}   
\end{eqnarray}
where $\tilde{\tilde{\bm{y}}}_{k,\mathcal{A}} \doteq  \bm{y}_{k,\mathcal{A}} + \bm{T}_\mathcal{A}^{-1} \bm{b}_k$, $k \geq t$.

It is important to note that,   for $k\geq t$, the estimator does not observe $\bm{y}_{k,\mathcal{A}}$. However, after attack, the estimator can calculate   $\tilde{\tilde{\bm{y}}}_{k,\mathcal{A}}$ from \eqref{eqn:equivalent-observation-unsafe-sensors},  since the estimator knows  $\tilde{\bm{z}}_{k,\mathcal{A}}$. Hence, we can define a new observation model for the estimator as follows: 
\begin{eqnarray}
    \tilde{\tilde{\bm{y}}}_{k,\mathcal{S}} &=& \bm{C}_\mathcal{S} \bm{x}_{k} + \bm{v}_{k,\mathcal{S}} = \bm{y}_{k,\mathcal{S}} \nonumber \\ 
    \tilde{\tilde{\bm{y}}}_{k,\bm{A}} &=& \bm{C}_\mathcal{A} \bm{x}_{k} + \underbrace{ \bm{v}_{k,\mathcal{A}} + \bm{T}_\mathcal{A}^{-1} \bm{b}_k}_{\doteq \tilde{\tilde{\bm{v}}}_{k,\mathcal{A}}} \nonumber\\ 
\implies \tilde{\tilde{\bm{y}}}_k   &=& \bm{C}  \bm{x}_{k} + \tilde{\tilde{\bm{v}}}_k 
\end{eqnarray}
where the observation noise from sensors belonging to $\mathcal{A}$ becomes  $\tilde{\tilde{\bm{v}}}_{k,\mathcal{A}} \doteq \bm{v}_{k,\mathcal{A}} + \bm{T}_\mathcal{A}^{-1} \bm{b}_k$, 
$\tilde{\tilde{\bm{y}}}_k \doteq \begin{bmatrix}
           \tilde{\tilde{\bm{y}}}_{k,\mathcal{S}} \\
            \tilde{\tilde{\bm{y}}}_{k,\mathcal{A}}
         \end{bmatrix}$, $\tilde{\tilde{\bm{v}}}_k \doteq  \begin{bmatrix}
           \tilde{\tilde{\bm{v}}}_{k,\mathcal{S}} \\
            \tilde{\tilde{\bm{v}}}_{k,\mathcal{A}}
         \end{bmatrix}$, and the covariance matrix of  $\tilde{\tilde{\bm{v}}}_k$ becomes 
$\tilde{\tilde{\bm{R}}} \doteq \begin{bmatrix}
           \bm{R}_{\mathcal{S}} \,\,\,\,\,\,\,\,\,\,\,\,\,\,\,\,\,\,\,\,\,\,\,\,\,\,\,\,\,\,\,\,\,\,\,\,\, \,\,\,\,\,\,\,\,\,\, \bm{0} \\
            \bm{0} \,\,\,\,\,\,\,\,\, \bm{R}_\mathcal{A} + \bm{T}_\mathcal{A}^{-1}\bm{\Sigma}_b\bm{T}_\mathcal{A}'^{-1}
         \end{bmatrix}$.         
         
It is to be noted that the estimator can use this model only for $k \geq t$, when it knows that event $\mathrm{a}$ is true, i.e., an attack has already happened. Basically, for $k \geq t$, the estimator can compute $\tilde{\bm{x}}_{k | \mathrm{a}}$ by using a standard Kalman filter under this new observation model. In this connection, we would also like to point out the connection between the innovation for this modified observation model and the innovation for the original observation model, both under attack:

\begin{eqnarray}\label{eqn:equivalent-innovation-in-terms-of-modified-innovation}
    \tilde{\tilde{\bm{z}}}_{k,\mathcal{S}} &=& \bm{y}_{k,\mathcal{S}} - \bm{C}_\mathcal{S}\bm{A}\tilde{\bm{x}}_{k-1|\mathrm{a}} \nonumber \\
    &=& \bm{y}_{k,\mathcal{S}} - \bm{C}_\mathcal{S}\bm{A}\tilde{\bm{x}}_{k-1} + \bm{C}_\mathcal{S}\bm{A}(\tilde{\bm{x}}_{k-1} - \tilde{\bm{x}}_{k-1|\mathrm{a}}) \nonumber \\
    &=& \tilde{\bm{z}}_{k,\mathcal{S}}+ \bm{C}_\mathcal{S}\bm{A}(\tilde{\bm{x}}_{k-1} - \tilde{\bm{x}}_{k-1|\mathrm{a}}) \nonumber\\ 
    \tilde{\tilde{\bm{z}}}_{k,\mathcal{A}} &=& \tilde{\tilde{\bm{y}}}_{k,\mathcal{A}} -  \bm{C}_\mathcal{A}\bm{A}\tilde{\bm{x}}_{k-1|\mathrm{a}} \nonumber \\
    &=& \bm{T}_\mathcal{A}^{-1}\tilde{\bm{z}}_{k,\mathcal{A}} + \bm{C}_\mathcal{A}\bm{A}(\tilde{\bm{x}}_{k-1} - \tilde{\bm{x}}_{k-1|\mathrm{a}})
\end{eqnarray}
Now, $\tilde{\bm{x}}_{k|\mathrm{a}}$ can be computed via using standard Kalman filter equations:
\begin{eqnarray}
    \tilde{\tilde{\bm{P}}}_{k|k-1} &=& \bm{A} \tilde{\tilde{\bm{P}}}_{k-1} \bm{A}'+\bm{Q} \nonumber\\
    \tilde{\tilde{\bm{K}}}_{k} &=& \tilde{\tilde{\bm{P}}}_{k | k-1} \bm{C}' (\bm{C} \tilde{\tilde{\bm{P}}}_{k|k-1} \bm{C}' +\tilde{\tilde{\bm{R}}})^{-1} \nonumber\\
    \tilde{\bm{x}}_{k|\mathrm{a}} &=& \bm{A}\tilde{\bm{x}}_{k-1|\mathrm{a}} + \tilde{\tilde{\bm{K}}}_{k}\begin{bmatrix}
           \tilde{\tilde{\bm{z}}}_{k,\mathcal{S}} \\
            \tilde{\tilde{\bm{z}}}_{k,\mathcal{A}}
         \end{bmatrix} \nonumber \\
    \tilde{\tilde{\bm{P}}}_{k} &=& (\bm{I}-\tilde{\tilde{\bm{K}}}_{k} \bm{C}) \tilde{\tilde{\bm{P}}}_{k|k-1} 
\end{eqnarray}

Note that, since this modified observation model makes sense only for $k \geq t$, the distribution of  $\tilde{\tilde{\bm{z}}}_k$ for $k\geq t$   depends on $t$.

\subsection{Computing the conditional  distribution of $\tilde{\bm{z}}_k$}  \label{distribution_calc}
From \eqref{eqn:equivalent-innovation-in-terms-of-modified-innovation}, we can write :
\begin{eqnarray}
    \tilde{\bm{z}}_{k,\mathcal{S}} &=& \tilde{\tilde{\bm{z}}}_{k,\mathcal{S}} + \bm{C}_\mathcal{S}\bm{A}(\tilde{\bm{x}}_{k-1|a} - \tilde{\bm{x}}_{k-1})
\end{eqnarray}
The innovation sequences  $\tilde{\tilde{\bm{z}}}_{k,\mathcal{S}}$ and $\tilde{\tilde{\bm{z}}}_{k,\mathcal{A}}$ are   Gaussian noise with   zero mean and covariance $\bm{C}_\mathcal{S}\tilde{\tilde{\bm{P}}}_{k|k-1}\bm{C}_\mathcal{S}' + \tilde{\tilde{\bm{R}}}_\mathcal{S}$ and $\bm{C}_\mathcal{A}\tilde{\tilde{\bm{P}}}_{k|k-1}\bm{C}_\mathcal{A}' + \tilde{\tilde{\bm{R}}}_\mathcal{A}$ respectively. Hence, the conditional  distribution of $\tilde{\bm{z}}_{k,\mathcal{S}}$ will be $ N(\bm{C}_\mathcal{S}\bm{A}(\tilde{\bm{x}}_{k-1|a} - \tilde{\bm{x}}_{k-1}),\bm{C}_\mathcal{S}\tilde{\tilde{\bm{P}}}_{k|k-1}\bm{C}_\mathcal{S}' + \tilde{\tilde{\bm{R}}}_\mathcal{S})$ given all necessary quantities involved, where $\tilde{\tilde{\bm{R}}}_\mathcal{S} = \bm{R}_\mathcal{S}$. Using \eqref{eqn:equivalent-innovation-in-terms-of-modified-innovation}, conditional distribution of $\tilde{\bm{z}}_{k,\mathcal{A}}$ is   $ N(\bm{T}_\mathcal{A}\bm{C}_\mathcal{A}\bm{A}(\tilde{\bm{x}}_{k-1|a} - \tilde{\bm{x}}_{k-1}),\bm{T}_\mathcal{A}(\bm{C}_\mathcal{A}\tilde{\tilde{\bm{P}}}_{k|k-1}\bm{C}_\mathcal{A}' + \tilde{\tilde{\bm{R}}}_\mathcal{A})\bm{T}_\mathcal{A}')$, where $\tilde{\tilde{\bm{R}}}_\mathcal{A} = \bm{R}_\mathcal{A} + \bm{T}_\mathcal{A}^{-1}\bm{\Sigma}_b\bm{T}_\mathcal{A}'^{-1}$. Under stability, $\tilde{\tilde{\bm{P}}}_{k|k-1}$ will converge to a limit  $\tilde{\tilde{\bm{P}}}$.

\section{Attack detection  for known linear attack: Bayesian setting}\label{section:known-attack-scheme}
In Section~\ref{section:recovering}, we have shown that the distribution of the innovation changes under a linear attack. In this section, we provide a quickest change  detection algorithm for detecting a linear attack with known $\bm{T}$ matrix. Motivated by the theory of \cite{poor2008quickest}, we formulate the problem as a partially observable Markov decision process (POMDP, see \cite[Chapter~$5$]{bertsekas07dynamic-programming-optimal-control-1}) and derive the optimal attack detection policy.

\subsection{Problem statement}
We assume that  the discrete time starts at $k=0$, and the attack begins at a random time~$t$ which is modeled as a geometrically distributed random variable with mean $\frac{1}{\theta}$, where $\theta \in (0,1)$ is known to the remote estimator. Clearly, $\theta$ denotes the probability that, given that no attack has started up to time~$k$, the attacker launches an attack at time~$(k+1)$. This assumption allows us to formulate the quickest attack detection problem as a POMDP.

The two hypotheses considered here are the following:

\noindent $H_0$: there is no attack.\\
$H_1$: there is an attack.

At each $k \geq 0$, the remote estimator computes $\bm{z}_k$ (or $\tilde{\bm{z}}_k$ provided that there is an attack). At some (possibly random) time~$\tau$, the estimator decides to stop  collecting observations $\bm{y}_k$ (or $\tilde{\bm{y}}_k$ provided that there is an attack), and declares that an attack has been launched on $\mathcal{A}$. This stopping time $\tau$ is determined by a policy (a sequence of decision rules)  $\bm{\mu}=\{\mu_k\}_{k \geq 0}$, where $\mu_k$ is a function that takes the history of observations available to the estimator at time~$k$ and decides whether to declare that an attack has been launched.  Clearly, $\tau < t$ denotes the event of a false alarm. 

We seek to minimize the expected delay in detecting an attack subject to a constraint on the false alarm probability: 
\begin{eqnarray}
    \min_{_{\bm{\mu}}} \mathbb{E}_{\bm{\mu}}[(\tau-t)^{+}]  
    \mbox{\,\,\, s.t. \,\,\,} \mathbb{P}_{\bm{\mu}} (\tau < t) \leq \alpha \label{eqn:constrained-problem}
\end{eqnarray}
This constrained problem can be relaxed using a Lagrange multiplier $\lambda>0$ to obtain the following unconstrained problem:
\begin{eqnarray}\label{eqn:unconstrained-problem}
    L(\bm{\mu}) = \mathbb{E}_{\bm{\mu}}[(\tau-t)^+ + \lambda \bm{1}_{\{\tau<t\}}]
\end{eqnarray}
The following standard results tells us how to choose $\lambda$:
\begin{theorem}\label{theorem:finding-optimal-lambda}
Let us consider \eqref{eqn:constrained-problem} and its relaxed version \eqref{eqn:unconstrained-problem}. If there exists a $\lambda^* \geq 0$ and  a policy $\bm{\mu}^*(\lambda^*)$ such that, (i)  $\bm{\mu}^*(\lambda^*)$ is an optimal policy for \eqref{eqn:constrained-problem} under $\lambda^*$, and (ii)  the constraint in \eqref{eqn:unconstrained-problem} is met with equality under $\bm{\mu}^*(\lambda^*)$, then $\bm{\mu}^*(\lambda^*)$ is an optimal policy for the constrained problem~\eqref{eqn:constrained-problem}.
\end{theorem}

\subsection{POMDP formulation} \label{Section:POMDP-formulation}
Let us define the belief probability of the estimator that an attack has been launched at or before time~$k$, as $\pi_k = \mathbb{P}(t\leq k|\tilde{\bm{z}}_1,\tilde{\bm{z}}_2,\ldots,\tilde{\bm{z}}_k)$, where  $\tilde{\bm{z}}_k=\bm{z}_k$ if $t>k$. Under this notation, \eqref{eqn:unconstrained-problem} can be rewritten as:

\footnotesize
\begin{eqnarray}
    L(\bm{\mu}) 
    =\mathbb{E}_{\bm{\mu}}[\sum_{k=0}^{\tau-1}\bm{1}_{(t\leq k)} + \lambda\bm{1}_{\{\tau<t\}}] 
    =\mathbb{E}_{\bm{\mu}}[\sum_{k=0}^{\tau-1}\pi_k + \lambda(1-\pi_\tau)] \label{eqn:delay-to-sum} 
\end{eqnarray}
\normalsize

For change detection, typically $\pi_k$ is a sufficient statistic \cite{l1983theory} for decision-making at time~$k$, but that does not hold in our problem due to temporal dependence of the innovation after an attack is  launched. Hence, we formulate a  POMDP with state at time~$k$  given by $(\pi_k,\tilde{\bm{z}}_1, \tilde{\bm{z}}_2,\ldots,\tilde{\bm{z}}_{k})$, with the understanding that  $\tilde{\bm{z}}_{k}=\bm{z}_k$ if $t>k$. The set of possible control actions is given by $\mathcal{U}=\{0,1\}$, where action~$0$ represents continuing to collect observations, and action~$1$ stands for stopping and declaring $H_1$.

Further, based on \eqref{eqn:delay-to-sum}, the single stage cost at time~$k$ is:
\[   
c_k(\pi_k,\tilde{\bm{z}}_1, \tilde{\bm{z}}_2,\ldots,\tilde{\bm{z}}_{k},u_k) = 
     \begin{cases}
       \pi_k &\quad\text{if }u_k=0\\
       \lambda(1-\pi_k) &\quad\text{if }u_k=1 \\ 
     \end{cases}
\]

\subsection{Recursive calculation of $\pi_k$} \label{subsection:recursiveCalculationOfPi}
We need to compute $\pi_k$ from $\pi_{k-1}$ recursively, in order to be able to calculate state transitions. However, after attack,  $\tilde{\bm{z}}_{k}$ sequence ceases to be i.i.d. across $k$, which makes this recursive calculation non-trivial. 
Note that, the joint probability density of the innovations computed at the estimator:
\begin{eqnarray}
    &&p(\tilde{\bm{z}}_1,\ldots,\tilde{\bm{z}}_k) \nonumber\\
    &=&p(\tilde{\bm{z}}_1)  p(\tilde{\bm{z}}_2|\tilde{\bm{z}}_1) \ldots p(\tilde{\bm{z}}_k|\tilde{\bm{z}}_1,\ldots,\tilde{\bm{z}}_{k-1})  \nonumber
\end{eqnarray}
Let us define  $p_c(\tilde{\bm{z}}_i) \doteq p(\tilde{\bm{z}}_i|\mathcal{H}_{i-1})$. Now, using Bayes rule, we can write:

\footnotesize
\begin{eqnarray}
    \pi_k &=& \frac{p(\tilde{\bm{z}}_1,\ldots,\tilde{\bm{z}}_{k}|t\leq k)\times \mathbb{P}(t\leq k)}{p(\tilde{\bm{z}}_1, \ldots,\tilde{\bm{z}}_k)} \nonumber \\
    &=& \frac{ \mathbb{P}(t\leq k)\Pi_{i=1}^{k}p_c(\tilde{\bm{z}}_i|t\leq k)}{\mathbb{P}(t\leq k)\Pi_{i=1}^{k}p_c(\tilde{\bm{z}}_i|t\leq k)+\mathbb{P}(t> k)\Pi_{i=1}^{k}p_c(\tilde{\bm{z}}_i|t> k)} \nonumber\\
    &=& \frac{\beta_k}{\beta_k+1} \label{eqn:pi_k-in-terms-of-beta_k}
\end{eqnarray}
\normalsize

where,

\footnotesize
\begin{eqnarray}
    \beta_k &=& \frac{\mathbb{P}(t\leq k)\Pi_{i=1}^{k}p_c(\tilde{\bm{z}}_i|t\leq k)}{\mathbb{P}(t> k)\Pi_{i=1}^{k}p_c(\tilde{\bm{z}}_i|t>k)}\nonumber \\
    &=&f_k(\theta)\frac{\mathbb{P}(t\leq k-1)p_c(\tilde{\bm{z}}_k|t\leq k)\Pi_{i=1}^{k-1}p_c(\tilde{\bm{z}}_i|t\leq k)}{\mathbb{P}(t> k-1)p_c(\tilde{\bm{z}}_k|t> k)\Pi_{i=1}^{k-1}p_c(\tilde{\bm{z}}_i|t>k-1)} \nonumber\\
    \label{eqn:BetaExpanded}
\end{eqnarray}
\normalsize

and, 
\begin{equation}
    f_k(\theta)=\frac{1-(1-\theta)^{k}}{(1-\theta)(1-(1-\theta)^{k-1})} \label{EqnTheta}
\end{equation}
Now, note that:
\begin{eqnarray}
    &&\Pi_{i=1}^{k-1}p_c(\tilde{\bm{z}}_i|t\leq k)\nonumber\\
    &=& p(\tilde{\bm{z}}_1,\ldots,\tilde{\bm{z}}_{k-1}|t\leq k) \nonumber\\
    &=& p(\tilde{\bm{z}}_1,\ldots,\tilde{\bm{z}}_{k-1}|t\leq k-1)\mathbb{P}(t\leq k-1|t\leq k) \nonumber\\
    && + p(\tilde{\bm{z}}_1,\ldots,\tilde{\bm{z}}_{k-1}|t= k)\mathbb{P}(t=k|t\leq k) \nonumber\\
    &=&  \frac{1-(1-\theta)^{k-1}}{1-(1-\theta)^k}  p(\tilde{\bm{z}}_1,\ldots,\tilde{\bm{z}}_{k-1}|t\leq k-1) \nonumber\\
    && +  \frac{\theta(1-\theta)^{k-1}}{1-(1-\theta)^k}  \Pi_{i=1}^{k-1}{p( \bm{z}_i)} \label{EqnKtoK-1} \label{eqn:recursive-calculation-of-the-numerator-of-beta_k}
\end{eqnarray}
where the product form in the last line comes from the fact that innovations are i.i.d. before the attack is launched. Clearly,  \eqref{eqn:recursive-calculation-of-the-numerator-of-beta_k} allows us to calculate the numerator of \eqref{eqn:BetaExpanded}. Similarly, in order to calculate the denominator of \eqref{eqn:BetaExpanded}, we note that:
\begin{equation}
    \Pi_{i=1}^{k-1}p_c(\tilde{\bm{z}}_i|t>k-1)=\Pi_{i=1}^{k-1} p(\bm{z}_i) \label{eqn:simplifying-denominator-of-beta_k}
\end{equation}
This again holds because innovations are i.i.d. before the attack is launched. 

Using \eqref{eqn:recursive-calculation-of-the-numerator-of-beta_k} and \eqref{eqn:simplifying-denominator-of-beta_k} in \eqref{eqn:BetaExpanded} and upon simplification, we obtain:

\footnotesize
\begin{eqnarray}
    \beta_k
    &=& \frac{p_c(\tilde{\bm{z}}_k|t\leq k)}{p_c(\tilde{\bm{z}}_k|t> k)} \bigg(\frac{\beta_{k-1}}{(1-\theta)}+\frac{\theta}{(1-\theta)} \bigg) \nonumber\\
     &=& \frac{p_c(\tilde{\bm{z}}_k|t\leq k)}{p_c(\tilde{\bm{z}}_k|t> k)} \bigg(\frac{\pi_{k-1}}{(1-\theta)(1-\pi_{k-1})}+\frac{\theta}{(1-\theta)} \bigg) \label{eqn:beta-update}
\end{eqnarray}
\normalsize

Obviously, $\pi_k$ can be calculated from $\beta_k$ using \eqref{eqn:pi_k-in-terms-of-beta_k}.

It is important to note that, using the results from Section~\ref{distribution_calc},  $p_c(\tilde{\bm{z}}_k|t\leq k)$ can be computed as:
\footnotesize
\begin{equation}
    p_c(\tilde{\bm{z}}_k|t\leq k)=\sum_{i=0}^{k}{p_c(\tilde{\bm{z}}_k|t=i)}\mathbb{P}_c(t=i|t\leq k) \\
\end{equation}
\begin{eqnarray}
    \mathbb{P}_c(t=i|t\leq k)&=&\frac{p(\tilde{\bm{z}}_1,\ldots,\tilde{\bm{z}}_{k-1}|t=i)\mathbb{P}(t=i|t\leq k)}{\substack{
    k\\
    \sum \\
    j=0}{p(\tilde{\bm{z}}_1,\ldots,\tilde{\bm{z}}_{k-1}|t=j) \mathbb{P}(t=j|t\leq k)}} \nonumber \\
    p(\tilde{\bm{z}}_1,\ldots,\tilde{\bm{z}}_{k-1}|t=i)&=&\Pi_{j=1}^{k-1}p_c(\tilde{\bm{z}}_j|t=i) \nonumber
\end{eqnarray}
\normalsize

Here $p_c(\tilde{\bm{z}}_k|t=i)$ is basically the distribution of $\tilde{\bm{z}}_k$ which has already been calculated in Section~\ref{distribution_calc} assuming that  $t=i$. However, for each $i \in \{0,1,\ldots,k\}$, we need to run a separate Kalman filter to calculate $p_c(\tilde{\bm{z}}_k|t=i)$, as described in \ref{distribution_calc}.   On the other hand, $p_c(\tilde{\bm{z}}_k|t> k)$ is simply the unconditional distribution of $\bm{z}_k$ under no attack, since the innovations are independent of each other under no attack. Hence, results from Section~\ref{distribution_calc} can be  directly used to calculate $\beta_k$ and hence $\pi_k$ recursively.

\subsection{Bellman equation, value function and policy structure}
Since the state transition is dependent on observation history, the optimal policy will be non-stationary, and the optimal value function will also be dependent on time.  Let $J_k^*(\pi_k,\tilde{\bm{z}}_1, \tilde{\bm{z}}_2,\ldots,\tilde{\bm{z}}_{k})$ denote the optimal cost-to-go starting from a state $(\pi_k,\tilde{\bm{z}}_1, \tilde{\bm{z}}_2,\ldots,\tilde{\bm{z}}_{k})$ at time~$k$, with $J^*(\pi_0) \doteq J_0^*(\pi_0)$. Also, let $\Psi_k$ be a function such that, given $u_k=0$, we have $\pi_{k+1}=\Psi_k(\pi_k, \tilde{\bm{z}}_1, \tilde{\bm{z}}_2, \cdots, \tilde{\bm{z}}_k)$. Hence,

The Bellman equation for \eqref{eqn:unconstrained-problem} is given by:
\begin{eqnarray}
   && J^{*}_k(\pi_k,\tilde{\bm{z}}_1, \tilde{\bm{z}}_2,\ldots,\tilde{\bm{z}}_{k}) \nonumber\\
   &=& \min\{\lambda(1-\pi_k),\pi_k+ \nonumber\\
   && \mathbb{E}[J^{*}_{k+1}(\Psi_k(\pi_k,\tilde{\bm{z}}_1, \tilde{\bm{z}}_2,\ldots,\tilde{\bm{z}}_{k}), \tilde{\bm{z}}_1, \tilde{\bm{z}}_2,\ldots,\tilde{\bm{z}}_{k+1} )| \tilde{\bm{z}}_{1:k}]\} \nonumber\\
   \label{eqn::Bellman-equation}
\end{eqnarray}
 The first term in the minimization of \eqref{eqn::Bellman-equation} is the cost  of stopping and declaring $H_1$, and the second term is the cost of continuing observation, which involves a single-stage cost $\pi_k$ and an expected cost-to-go from the next step where the expectation is taken over the distribution of $\tilde{\bm{z}}_{k+1}$ conditioned on $\tilde{\bm{z}}_{1:k} \doteq (\tilde{\bm{z}}_1, \ldots,\tilde{\bm{z}}_{k})$.

Let us consider an $N$-horizon problem which is same as problem~\eqref{eqn:unconstrained-problem} except that, at time~$N$, the detector must stop and declare $H_1$. Let the analogues of $J_k^*$ (for various values of $k$) and $J^*$ for this $N$-horizon problem be denotes by $J_k^{(N)*}$ and $J^{(N)*}$, respectively.

\begin{lemma}\label{lemma:truncated-value-function-concave-in-pi}
$J^{(N)*}(\pi)$ is concave in $\pi \in [0,1]$.
\end{lemma}
\begin{proof}
See Appendix~\ref{appendix:proof-of-truncated-value-function-concave-in-pi}. Proof is based on outline from \cite{premkumar2008Sleepwake}.
\end{proof}

\begin{theorem}
 $J^{*}(\pi)$ is concave in $\pi$.
\end{theorem}
\begin{proof}
The proof follows from Lemma~\ref{lemma:truncated-value-function-concave-in-pi} and the fact that limit of a sequence of concave functions is concave.
\end{proof} 

The next theorem describes the optimal policy for the unconstrained problem~\eqref{eqn:unconstrained-problem}.
\begin{theorem}\label{theorem:threshold-policy-structure}
 The optimal  policy for the constrained  problem~\eqref{eqn:unconstrained-problem} is a threshold policy. At time~$k$, if the state is $(\pi_k,\tilde{\bm{z}}_1, \tilde{\bm{z}}_2,\ldots,\tilde{\bm{z}}_{k})$, then the optimal action is to stop and declare $H_1$ if $\pi_k>\Gamma_{k}(\tilde{\bm{z}}_{1:k})$ for a threshold $\Gamma_{k}(\tilde{\bm{z}}_{1:k}) \in [0,1]$, and to continue collecting observations if $\pi_k<\Gamma_{k}(\tilde{\bm{z}}_{1:k})$. If $\pi_k=\Gamma_{k}(\tilde{\bm{z}}_{1:k})$, either action is optimal.
\end{theorem} 
\begin{proof}
See Appendix~\ref{appendix:proof-of-threshold-policy-structure}.
\end{proof}

\subsection{Computational complexity and the QUICKDET algorithm}
Due to the time-inhomogeneous state transition, problem~\eqref{eqn:unconstrained-problem} cannot be solved by standard techniques such as value iteration. On the other hand, the optimal threshold $\Gamma_{k}(\tilde{\bm{z}}_{1:k})$ at time $k$  depends not only on $k$ but also the history of innovations $\tilde{\bm{z}}_{1:k}$. Further, the presence of  $J_{k+1}^*$ in the Bellman equation~\eqref{eqn::Bellman-equation}  makes even a fairly discretized version of the problem computationally very heavy to solve in an online fashion. These problems can be alleviated by setting a constant threshold, i.e., $\Gamma_{k}(\tilde{\bm{z}}_{1:k})=\Gamma$ for all $k \geq 1, \tilde{\bm{z}}_{1:k}$. Definitely, this will yield a suboptimal solution, but as will be seen later, we can achieve much better performance than the competing algorithms in the literature by carefully choosing $\Gamma$.

\subsubsection{The QUICKDET algorithm}
Motivated by Theorem~\ref{theorem:threshold-policy-structure} along with the need for a constant threshold $\Gamma$ to reduce computational complexity, here we describe outline of an algorithm QUICKDET. The two key components of QUICKDET, apart from the threshold structure, are the choices of the optimal $\Gamma^*$ to minimize the objective in the unconstrained problem~\eqref{eqn:unconstrained-problem} within the class of stationary threshold policies, and  $\lambda^*$ to meet the constraint in \eqref{eqn:constrained-problem} with equality as per Theorem~\ref{theorem:finding-optimal-lambda}.  These parameters are set via some off-line pre-computation involving two timescale stochastic approximation \cite{borkar08stochastic-approximation-book}, wherein  we update $\lambda$ in the slower timescale, and $\Gamma$ in the faster timescale. 

In the pre-computation phase, we generate various sample paths $\mathcal{P}_0, \mathcal{P}_1, \cdots$ of the process, observation and attack dynamics as per our discussed system model. The iterations start with some initial iterates $\Gamma(0)$ and $\lambda(0)$. For sample path~$\mathcal{P}_n$, a threshold policy with constant threshold $\Gamma(n)$ is used for attack detection on the state and modified observation sequence, and it is observed whether this policy generates a false alarm for sample path~$\mathcal{P}_n$; if it does not generate a false alarm, then the detection delay $\tau_n$ is recorded. Let $\mathbbm{1}_{FA}(n)$ and $\mathbbm{1}_{D}(n)$ be the indicators of false alarm and attack detection for sample path~$\mathcal{P}_n$. 

Let $\{a(n)\}_{n \geq 0}$, $\{b(n)\}_{n \geq 0}$ and $\{\delta(n)\}_{n \geq 0}$ be three non-negative sequences satisfying the following properties: (i) $\sum_{n=0}^{\infty} a(n)=\sum_{n=0}^{\infty} b(n)=\infty$, (ii)  $\sum_{n=0}^{\infty} a^2(n)< \infty, \sum_{n=0}^{\infty} b^2(n)<\infty$, (iii) $\lim_{n \rightarrow \infty}\frac{b(n)}{a(n)}=0$, (iv) $\lim_{n \rightarrow \infty} \delta(n)=0$, and (v) $\sum_{n=0}^{\infty} \frac{a^2(n)}{\delta^2(n)}<\infty$. The first two conditions are standard requirements for stochastic approximation. Condition~(iii) ensures the necessary timescale separation. The last two conditions are required for the convergence of the $\Gamma(n)$ update via stochastic gradient descent (SGD), adapted from the theory of simultaneous perturbation stochastic approximation (SPSA, see \cite{spall92original-SPSA}).

Let $\Gamma^+(n)=\Gamma(n) +\delta(n)$ and $\Gamma^-(n)=\Gamma(n) -\delta(n)$ be two perturbations of $\Gamma(n)$ in opposite directions. Let $d_n^+=\tau_n \mathbbm{1}_{D}(n)+\lambda(n) \mathbbm{1}_{FA}(n)|_{\Gamma^+(n)}$ be the cost incurred along sample path~$\mathcal{P}_n$ if a  threshold policy with a constant threshold $\Gamma^+(n)$ is used along sample path $\mathcal{P}_n$; let us define $d_n^-$ is a similar way.

The following updates are made: 
\begin{eqnarray}\label{eqn:policy-pre-computation}
\Gamma(n+1)&=&[\Gamma(n)-a(n)\times\frac{d_n^+ -d_n^-}{2\delta(n)}]_0^1 \nonumber \\
    \lambda(n+1)&=& [\lambda(n)+b(n)\times(\mathbbm{1}_{FA}(n)-\alpha)]_0^{\infty} 
\end{eqnarray}
The $\Gamma(n)$ update in \eqref{eqn:policy-pre-computation} is a stochastic gradient descent algorithm used to minimize the expected cost along sample path $\mathcal{P}_n$, for $\lambda=\lambda(n)$. This algorithm runs in a faster timescale. The $\lambda(n)$ update runs at a slower timescale to ensure that the false alarm probability equals $\alpha$. The faster timescale iterate $\Gamma(n)$ views the $\lambda(n)$ iterate as quasi-static, while the $\lambda(n)$ iterate views the $\Gamma(n)$ iterate as almost equilibriated. The iterates are projected onto desired intervals to ensure boundedness.

Using standard arguments, we can prove that $(\Gamma(n), \lambda(n)) \rightarrow \{(\Gamma, \lambda) \in [0,1] \times [0,\infty): \mathbb{E}(\mathbbm{1}_{FA})=\alpha, \nabla_{\Gamma}  \mathbb{E}(d)|_{\lambda}=0 \}$. In practice, if $(\Gamma(n), \lambda(n)) \rightarrow (\Gamma^*, \lambda^*)$, then this $\Gamma^*$ is used in the real attack detector on field.

\subsection{Extension to multiple attack matrices}
\subsubsection{Calculation of transition probabilities}
Here we assume that the attack matrix $\bm{T} $ is unknown but belongs to a known finite set  $\mathcal{T}$, and that $\bm{I} \in \mathcal{T}$. We assume an initial prior distribution on $\{ \bm{T} \in \mathcal{T}: \bm{T} \neq \bm{I} \}$.  The  detector maintains a belief probability $\pi_k^{\bm{T}}$ for hypothesis  $\bm{T}$, and the total number of hypotheses is $|\mathcal{T}|$. In this case, the posterior belief  probability of attack becomes $\pi_k=1-\pi_k^{\bm{I}}$. Here,  $\pi_k^{\bm{T}}$ for $\bm{T} \neq \bm{I}$ is defined as the belief at time~$k$ that the attack has already begun, and that the attacker uses $\bm{T}$:

\begin{eqnarray}
   \pi_k^{\bm{T}} & \doteq & p(t\leq k, \bm{T} | \tilde{\bm{z}}_1, \ldots, \tilde{\bm{z}}_k) \nonumber \\
   &=& \underbrace{p(t\leq k | \tilde{\bm{z}}_1, \ldots, \tilde{\bm{z}}_k, \bm{T})}_{\doteq \pi_{k|\bm{T}}} p(\bm{T} | \tilde{\bm{z}}_1, \ldots, \tilde{\bm{z}}_k) \nonumber
\end{eqnarray}
where $\pi_{k|\bm{T}} = \frac{\beta_{k|{\bm{T}}}}{\beta_{k|{\bm{T}}} + 1}$ (similar to \eqref{eqn:pi_k-in-terms-of-beta_k}) is the conditional belief on attack having already started  given that the attacker uses $\bm{T}$, and this $\beta_{k|{\bm{T}}}$ can be calculated in the same way as in Section~\ref{subsection:recursiveCalculationOfPi} for each $\bm{T}\in\mathcal{T}$. Posterior probability distribution over $\mathcal{T}$, i.e.,  $p(\bm{T}|\tilde{\bm{z}_1},\ldots,\tilde{\bm{z}_k}, t \leq k)$ for $\bm{T} \neq \bm{I}$, can be calculated as given below:

\footnotesize
\begin{eqnarray}
    && p(\bm{T}|\tilde{\bm{z}_1},\ldots,\tilde{\bm{z}_k}) \nonumber\\
    &=& \frac{p(\tilde{\bm{z}}_k|\tilde{\bm{z}}_1, \ldots, \tilde{\bm{z}}_{k-1}, \bm{T}) p(\bm{T}|\tilde{\bm{z}}_1,\ldots,\tilde{\bm{z}}_{k-1})}{\sum_{\bm{T}' \in \mathcal{\bm{T}}}p(\tilde{\bm{z}}_k|\tilde{\bm{z}}_1, \ldots, \tilde{\bm{z}}_{k-1}, \bm{T}') p(\bm{T}'|\tilde{\bm{z}}_1,\ldots,\tilde{\bm{z}}_{k-1})} \nonumber \\
    &=& \frac{p_c(\tilde{\bm{z}}_k|\bm{T})p(\bm{T}|\tilde{\bm{z}}_1,\ldots,\tilde{\bm{z}}_{k-1})}{\sum_{\bm{T}' \in \mathcal{\bm{T}}} p_c(\tilde{\bm{z}}_k|\bm{T}')p(\bm{T}'|\tilde{\bm{z}}_1,\ldots,\tilde{\bm{z}}_{k-1})} \nonumber
\end{eqnarray}
\normalsize

Here $p(\bm{T}|\tilde{\bm{z}}_1,\ldots,\tilde{\bm{z}}_{k-1})$ is known through the iterative calculation, where the first iteration will depend on prior distribution on attack matrices $\bm{T} \in \mathcal{T}$. The probability $p_c(\tilde{\bm{z}}_k| \bm{T})$ is the distribution of innovation given the history and attack matrix $\bm{T}$: 
\begin{eqnarray}
    p_c(\tilde{\bm{z}}_k| \bm{T}) &=& p_c(\tilde{\bm{z}}_k|\bm{T},t\leq k)p(t\leq k|\tilde{\bm{z}}_1,\ldots,\tilde{\bm{z}}_{k-1}, \bm{T}) + \nonumber \\
    & & p_c(\tilde{\bm{z}}_k|\bm{T},t > k)p(t > k|\tilde{\bm{z}}_1,\ldots,\tilde{\bm{z}}_{k-1},\bm{T}) \nonumber \\
    &=&  p_c(\tilde{\bm{z}}_k|\bm{T},t\leq k)(\pi_{k-1|\bm{T}} + (1-\pi_{k-1|\bm{T}})\theta) + \nonumber \\
    & &  p_c(\tilde{\bm{z}}_k|\bm{T},t > k)(1-\pi_{k-1|\bm{T}})(1-\theta) \label{eqn:m-ary-posterior-distribution}
\end{eqnarray}
Now, by using \eqref{eqn:beta-update}, we can write:
\begin{eqnarray}
    \beta_{k|{\bm{T}}} = \frac{p_c(\tilde{\bm{z}}_k|\bm{T}, t \leq k)(\pi_{k-1|{\bm{T}}} + (1-\pi_{k-1| {\bm{T}}})\theta)}{p_c(\tilde{\bm{z}}_k|\bm{T}, t > k) (1-\theta)(1-\pi_{k-1| {\bm{T}}})} \nonumber
\end{eqnarray}
Hence, 

\small
\begin{equation}
\begin{split}
    \beta_{k|{\bm{T}}} p_c(&\tilde{\bm{z}}_k|\bm{T}, t > k) (1-\theta)(1-\pi_{k-1|{\bm{T}}}) \\
    &\quad = p_c(\tilde{\bm{z}}_k|\bm{T}, t \leq k)(\pi_{k-1|{\bm{T}}} + (1-\pi_{k-1| {\bm{T}}})\theta) \label{eqn:beta-substitute-calc}
\end{split}
\end{equation}
\normalsize

Finally, \eqref{eqn:m-ary-posterior-distribution} can be simplified by substituting  \eqref{eqn:beta-substitute-calc} resulting the following expression:

\footnotesize
\begin{eqnarray}
     p_c(\tilde{\bm{z}}_k| \bm{T}) &=& p_c(\tilde{\bm{z}}_k|\bm{T},t > k)(1-\pi_{k-1| {\bm{T}}})(1+\beta_{k|{\bm{T}}})(1-\theta) \nonumber
\end{eqnarray}
\normalsize

Now, $p_c(\tilde{\bm{z}}_k|\bm{T},t > k)$ becomes the  unconditional distribution of the innovation under no attack. On the other hand, $\pi_{k-1|\bm{T}}$ is known from the previous iteration, and $\beta_{k|\bm{T}}$ can be calculated as  in Section~\ref{subsection:recursiveCalculationOfPi}. Thus, we can calculate $p_c(\tilde{\bm{z}}_k| \bm{T})$ and consequently $p(\bm{T}|\tilde{\bm{z}_1},\ldots,\tilde{\bm{z}_k})$ and $\pi_{k|\bm{T}}$ for all $\bm{T} \in \mathcal{T}$. 

However, the time complexity of algorithm in multiple attack matrix case  increases  $|\mathcal{T}|$ times compared to the binary hypothesis testing case, since we would need to maintain separate set of Kalman filters for each attack matrix $\bm{T}$.

\subsubsection{Policy structure}
For     \eqref{eqn:constrained-problem}, we can still compute the belief probability that the attack has begun:  
\begin{eqnarray}
   \pi_k &=& p(t\leq k|\tilde{\bm{z}}_1,\tilde{\bm{z}}_2,\ldots,\tilde{\bm{z}}_k) \nonumber \\
   &=& \sum_{\bm{T} \in \mathcal{T}}p(t\leq k, \bm{T}|\tilde{\bm{z}}_1,\tilde{\bm{z}}_2,\ldots,\tilde{\bm{z}}_k) \nonumber \\
   &=& \sum_{\bm{T} \in \mathcal{T}}\pi_k^{\bm{T}} \nonumber
\end{eqnarray}
Similar to Theorem~\ref{theorem:threshold-policy-structure}, the optimal Policy structure for multiple attack matrices  will still be a threshold policy. 
We also note that,  $\underset{\bm{T} \in \mathcal{T}}{\arg\max}  \pi_{k|\bm{T}}$ yields the maximum likelihood (ML) detection of $\bm{T}$, and  $ \underset{\bm{T} \in \mathcal{T}}{\arg\max}  \pi_k^{\bm{T}}$ can be used for maximum aposteriori (MAP) detection of the attacker's strategy.

\section{Attack detection in the non-Bayesian setting}
\label{section:quickdet-non-Bayesian}
In this section, we consider the situation where the distribution of the attack initiation instant $t$ is not known. This eliminates the possibility of solving the problem via MDP formulation.  Hence, we  use the popular generalised CUSUM algorithm proposed by Lai \cite{737522} for change detection, and demonstrate how  the results in Section~\ref{distribution_calc} can be used in computing the generalised CUSUM statistic.

  The basic CUSUM Algorithm was developed heuristically by Page \cite{10.1093/biomet/41.1-2.100}, but was later analysed rigorously in \cite{10.1214/aoms/1177693055}, \cite{10.1214/aos/1176350164}, \cite{10.1214/aos/1176347761} and \cite{737522}.  In our non-Bayesian FDI detection setting, under no attack (i.e., $t=\infty$), the false alarm rate (FAR)  is the inverse of the mean time to false alarm $FAR(\tau) = \frac{1}{\mathbb{E}_\infty(\tau)}$. 
With little abuse of notation where the stopping time  $\tau$ also represents a stopping rule for the detector, we focus on the following set of detection/stopping rules:
 \begin{equation}
     D_{\alpha} = \{\tau : FAR(\tau) \leq \alpha\} \label{eqn:stopping-Time-set}
 \end{equation}
Since finding a uniformly powerful test that minimizes detection delay over $D_{\alpha}$ is not possible, we study two important minimax formulations developed by Lorden \cite{10.1214/aoms/1177693055} and Pollak \cite{10.1214/aos/1176346587}.

 {\em Lordan's Problem:} For a given $\alpha$, find   $\tau \in D_{\alpha}$ to minimize worst average detection delay $WADD(\tau)$ defined as:
 \begin{equation}
     WADD(\tau) = sup_{n\geq 1} ess sup \mathbb{E}_{n}[(\tau - n)^{+} | \tilde{\bm{z}}_1, \ldots, \tilde{\bm{z}}_{n-1}] \label{eqn:WADD}
 \end{equation}
 
 {\em Pollak's Problem:} For a given $\alpha$, find   $\tau \in D_{\alpha}$ to minimize cumulative average detection delay $CADD(\tau)$ defined as :
 \begin{equation}
     CADD(\tau) = sup_{n\geq 1} \mathbb{E}_{n}[\tau - n | \tau \geq n] \label{eqn:CADD}
 \end{equation}
 It was proved in \cite[Section~IV]{veeravalli2012quickest} that $WADD(\tau) \geq CADD(\tau)$. Lai \cite{737522} studied both Lorden's and Pollak's problem in non-i.i.d. setting, and showed that under some additional condition, the basic CUSUM algorithm is asymptotically optimal as $\alpha \to  0$.

 Let us define $p_c(\tilde{\bm{z}_i} | t = k) $ (adapted to our problem setting where innovations are i.i.d. across time before attack),  the associated log-likelihood ratio, and the corresponding running sum of the log-likelihood ratios as follows:
 \begin{eqnarray}
 p_c(\tilde{\bm{z}}_i | t = k)  &\doteq & p(\tilde{\bm{z}}_i | \tilde{\bm{z}}_0.\cdots, \tilde{\bm{z}}_{i-1}, t = k)\nonumber\\
     L_{i,k} & \doteq & log \bigg(\frac{p_c(\tilde{\bm{z}}_i | t = k)}{p(\tilde{\bm{z}}_i | t =\infty)} \bigg) \nonumber\\
     S_n &\doteq & max_{1\leq k \leq n} \underbrace{\sum_{i=k}^{n} L_{i,k}}_{\doteq S_{n,k}}  
 \end{eqnarray}

   \begin{figure}[t!]
\begin{centering}
\begin{center}
\includegraphics[height=6.5cm, width=0.92\linewidth]{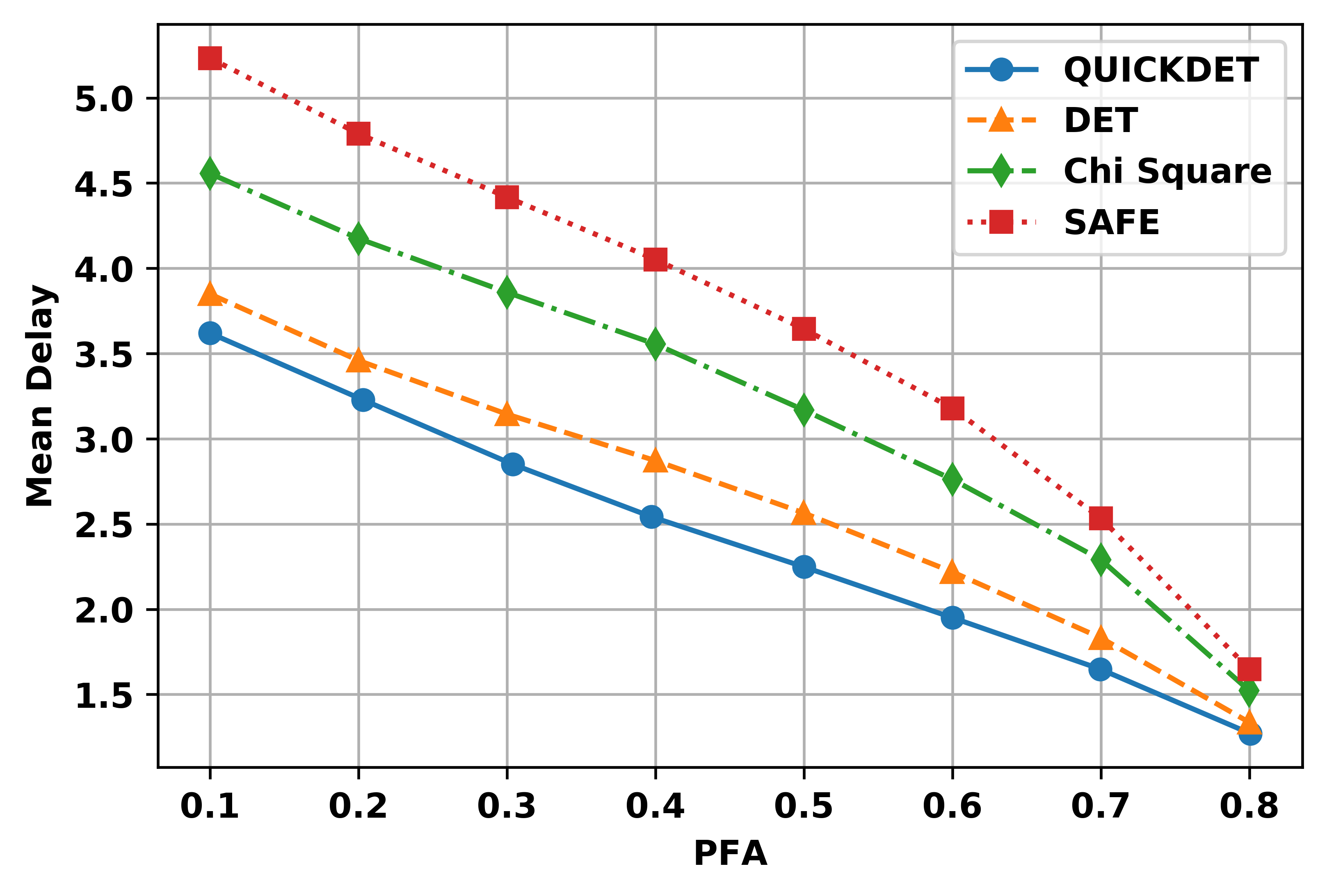}
\end{center}
\end{centering}
\vspace{-5mm}
\caption{Mean delay versus probability of false alarm (PFA)  comparison among QUICKDET and three other  detectors.}
\label{fig:Detection-Comparison}
\vspace{-2mm}
\end{figure}

 Here $S_{n,k}$ represents the cumulative sum of likelihood ratio up to time~$n$ assuming that the attack started at $t=k$. Obviously, $S_{n,k}$ can be written recursively as follows:
 \begin{equation}
     S_{n,k} = \max \{0, S_{n-1, k} + log \bigg( \frac{p_c(\tilde{\bm{z}}_n | t=k)}{p(\tilde{\bm{z}}_n | t=\infty)} \bigg) \} \nonumber
 \end{equation}
\begin{remark}
 It is important to note that, $p_c(\tilde{\bm{z}}_i| t=k)$ for $i \geq k$ can be computed by using the theory developed in Section~\ref{distribution_calc}.
\end{remark}

The generalized CUSUM algorithm   for  non i.i.d. observations involve the following stopping time: 
 \begin{equation}
     \tau_g = \inf \{ n \geq 1 : S_n \geq b\}
 \end{equation}
  for some threshold $b$.  It was proved in \cite{737522} that, as  $\alpha \to 0$, under some regularity conditions, 
  \begin{eqnarray}
      && \mathbb{E}_\infty[\tau_g] \geq e^b \nonumber\\
     && CADD(\tau_g) \leq WADD(\tau_g) \leq \frac{b}{I}(1+o(1)) \text{  as   }b \to \infty \nonumber
  \end{eqnarray}
  for a constant $I$. Hence, If we set $b=|log\alpha|$, then
  \begin{eqnarray}
      FAR(\tau_g) &=& \frac{1}{\mathbb{E}_\infty [\tau_g]} \leq  \alpha \nonumber\\
      WADD(\tau_g) &\leq & \frac{|log\alpha|}{I}(1+o(1))
  \end{eqnarray}
  Thus, $\tau_g$ is first-order asymptotically optimal detection rule  within   $D_\alpha$.

\begin{figure}[t!]
\begin{centering}
\begin{center}
\includegraphics[height=6.5cm, width=0.92\linewidth]{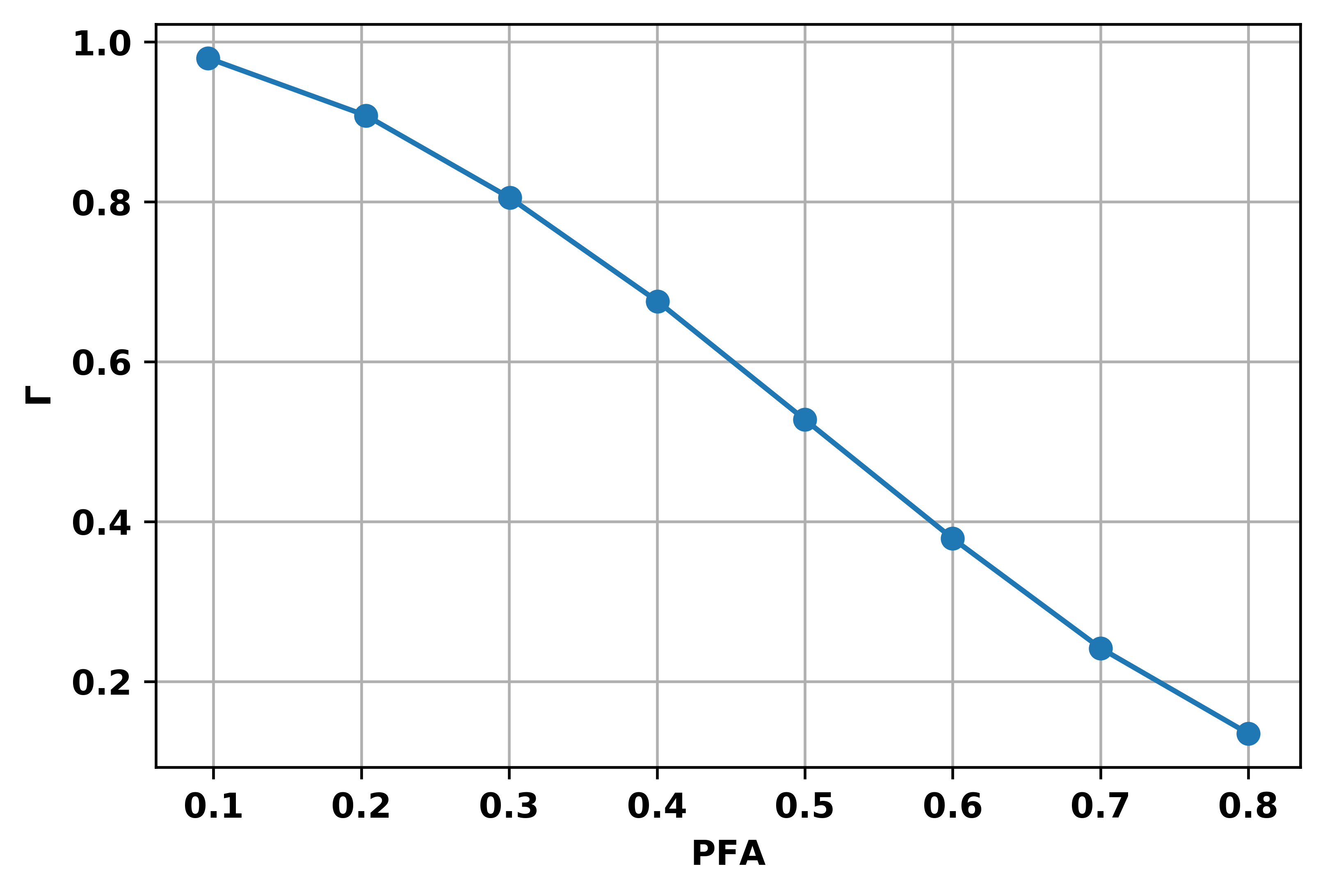}
\end{center}
\end{centering}
\vspace{-5mm}
\caption{Variation of threshold $\Gamma$ with probability of false alarm (PFA)  for   QUICKDET.}
\label{fig:threshold-vs-PFA}
\vspace{-5mm}
\end{figure}

\section{Numerical results}\label{section:numerical-results}
Since multiple sensor observations can be viewed as parts of a single imaginary sensor's observation, we assume that there is one safe sensor and one unsafe sensor, and the sign of the innovation coming from the unsafe sensor is inverted, i.e., $\bm{T}_{\mathcal{A}}=-\bm{I}$. 
In this section, we first compare the performance of QUICKDET against three other detectors:
\begin{itemize}
    \item {\bf $\chi^2$ detector:} This detector is as described in Section~\ref{section:system-model}, except that $\eta$ is optimized in an off-line pre-computation phase (as in QUICKDET) to meet the false alarm constraint with equality: 
    $$\eta(n+1)= [\eta(n)+a(n)\times(\mathbbm{1}_{FA}(n)-\alpha)]_0^{\infty} $$
    The limit $\eta^*$ of this iteration is used in the real detector on field. We choose $J=3$ in our simulation.
    \item {\bf DET:}  This is an adaptation of the DET algorithm in  \cite{chattopadhyay2019security}. It requires the detector  to run two separate parallel Kalman filters for the safe and unsafe sensors.  Let $\hat{\bm{x}}_{k,\mathcal{S}}$ and $\hat{\bm{x}}_{k,\mathcal{A}}$ be the estimates declared by two blind Kalman filters using observations from the safe sensor and from the unsafe sensor, respectively. This detector declares an attack at time~$j$ if $\sum_{k=j-J+1}^j (\hat{\bm{x}}_{k,\mathcal{A}}-\hat{\bm{x}}_{k,\mathcal{S}})' \bm{\Sigma}^{-1} (\hat{\bm{x}}_{k,\mathcal{A}}-\hat{\bm{x}}_{k,\mathcal{S}}) > \eta$ where $\eta$ can be optimized as in the $\chi^2$~detector to meet the false alarm constraint with equality, and $\bm{\Sigma}$ is the steady state covariance matrix of $(\hat{\bm{x}}_{k,\mathcal{A}}-\hat{\bm{x}}_{k,\mathcal{S}})$ under no attack. We choose $J=3$ in our simulation.
    \item {\bf SAFE:} This is the detection algorithm taken from  \cite{li2017detection}, with the threshold optimized as before to meet the false alarm constraint.
\end{itemize}

\begin{figure}[t!]
    \begin{centering}
    \begin{center}
        \includegraphics[height=6.5cm, width=0.92\linewidth]{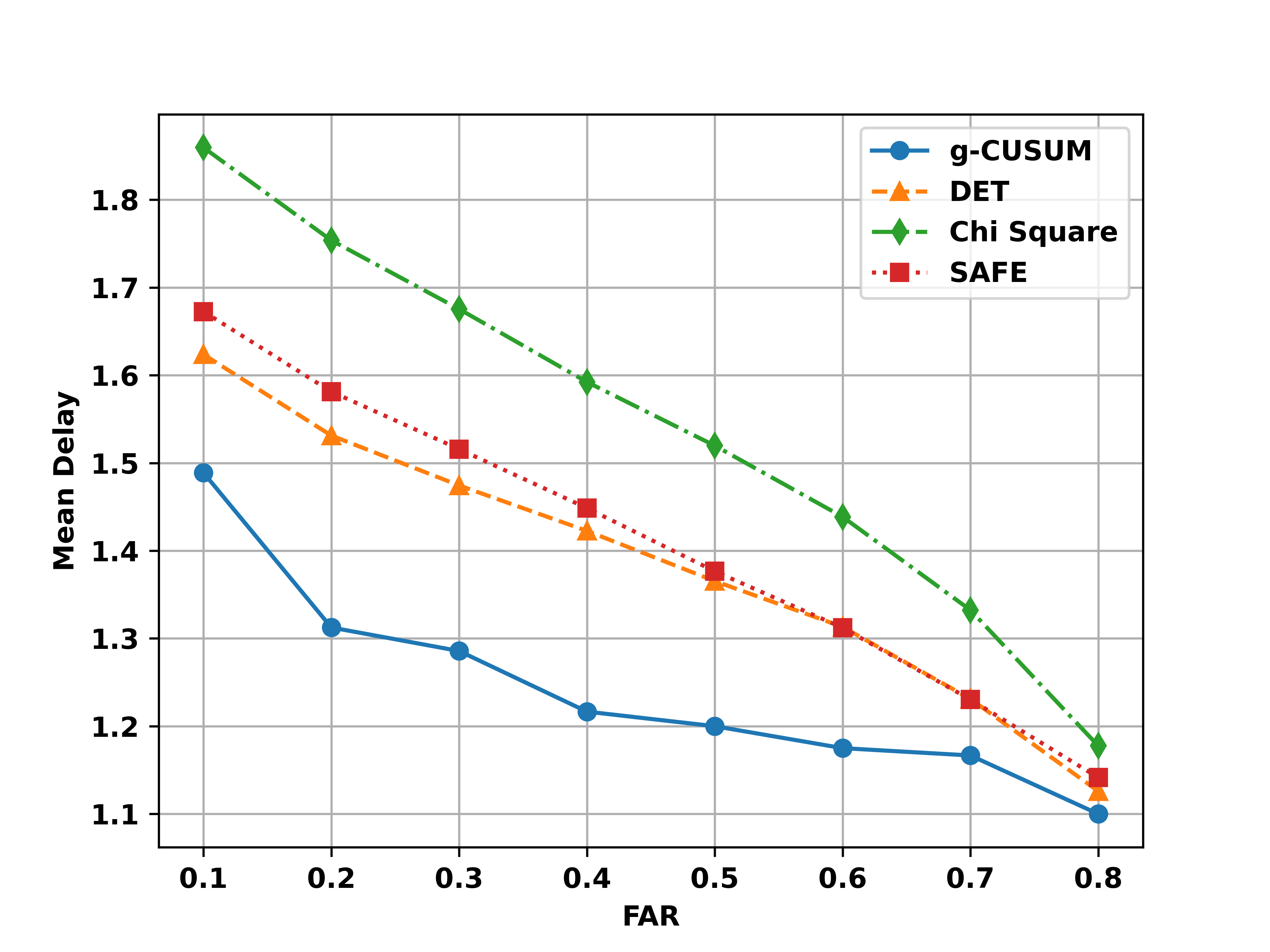}
    \end{center}
    \end{centering}
    \vspace{-5mm}
    \caption{Mean detection delay versus FAR performance comparison in the non-Bayesian setting.}
    \label{fig:CUSUM-Detection-Plot}
    \vspace{-5mm}
    \end{figure}

We consider a process with dimension $q=2$. The $\bm{A}$, $\bm{Q}$, $\bm{R}_{\mathcal{A}}$ and $\bm{R}_{\mathcal{S}}$ matrices are all chosen to be $\begin{bmatrix}
            1\,\,0  \\
            0\,\,1
         \end{bmatrix}$, $\bm{C}_{\mathcal{A}}=\begin{bmatrix}
            0.5\,\,0  \\
            0\,\,\,1
         \end{bmatrix}$ and $\bm{C}_{\mathcal{S}}=\begin{bmatrix}
            0\,\,0.5  \\
            1\,\,\,0
         \end{bmatrix}$ . The probability of launching an attack at a time, $\theta$, is chosen to be $0.05$, and we generate $10000$ sample paths for the pre-computation phase of QUICKDET. Prior belief probability of attack  $\pi_0$ is set to $0$ for all the sample paths. 

Figure~\ref{fig:Detection-Comparison} compares the mean detection delay of QUICKDET, DET, $\chi^2$ and SAFE detectors, for various values of false alarm probability $\alpha$. We observe that QUICKDET  outperforms all other detectors, and the performance margin is very high w.r.t. SAFE and $\chi^2$ detectors. It is important to note that, though QUICKDET is a suboptimal detector motivated by the optimal detector of Theorem~\ref{theorem:threshold-policy-structure}, it uses the knowledge of the matrix $\bm{T}_{\mathcal{A}}$ which the other three algorithms do not use. Hence, given the knowledge of $\bm{T}_{\mathcal{A}}$ (the attacker's strategy), it is always better to use QUICKDET. It is also observed that DET's performance is not far from that of QUICKDET, which means that DET can be used as a potential alternative to QUICKDET in case the knowledge of $\bm{T}_{\mathcal{A}}$ is not available apriori. These observations have been verified through numerous simulation experiments under various problem instances.

Since the probability of attack detection is $1$ in all of our simulations, we do not compare ROC curves of the detectors. However, detection delay comparison of Figure~\ref{fig:Detection-Comparison} is a reasonable alternative to the ROC plots.

Figure~\ref{fig:threshold-vs-PFA} shows that the threshold $\Gamma$ in QUICKDET decreases with the false alarm constraint $\alpha$, which supports the intuition that a larger $\Gamma$ results in less frequent up-crossing of $\Gamma$ by the belief probability, and consequently a smaller false alarm probability. 

Figure~\ref{fig:CUSUM-Detection-Plot} compares the performance of generalised CUSUM test with the other three algorithms in terms of mean detection delay versus FAR in the non-Bayesian detection setting. It is observed that here again generalised CUSUM (G-CUSUM) significantly outperforms the other three algorithms.

\section{Conclusion}\label{section:conclusion}
In this paper, we  provided an algorithm for quickest detection of FDI attack on remote estimation in the Bayesian and non-Bayesian setting. Theoretical proof of optimality was provided wherever required, and numerical results  demonstrated clear superiority of the proposed algorithms against other competing algorithms. However, there remain a number of open questions: (i) how to handle a nonlinear attack scheme? (ii) how to handle unknown attack strategy? (iii) how to perform distrbuted quickest detection in a multihop network setting? We plan to address these issues in our future research endeavours.

 \appendices 

 \section{Proof of Lemma~\ref{lemma:truncated-value-function-concave-in-pi}}
 \label{appendix:proof-of-truncated-value-function-concave-in-pi}
 
 Under action $u_{k-1}=0$ and for $\pi_{k-1}=\pi$, the recursion for $\pi$ can be written in a more elementary form as:
\begin{equation}
    \pi_{next}=\frac{p_c(\tilde{\bm{z}}_k|t\leq k)(\pi+(1-\pi)\theta)}{p(\tilde{\bm{z}}_k|\tilde{\bm{z}}_1,\tilde{\bm{z}}_2,\ldots,\tilde{\bm{z}}_{k-1})}
\end{equation}

Note that,  $J_N^{(N)*}(\pi)=\lambda(1-\pi)$ is concave in $\pi$. Now $J_{N-1}^{(N)*}(\pi)=\min\{\lambda(1-\pi),\pi+\mathbb{E}[J_N^{(N)*}(\pi_{next})]\}$ which is minimum of two linear functions and hence concave in $\pi$.

As induction hypothesis, let us assume that $J_k^{(N)*}(\pi)$ is  concave in $\pi$, and we would like to show that $J_{k-1}^{(N)*}(\pi)$ is concave in $\pi$. Since $J_{k-1}^{(N)*}(\pi)=\min \{\lambda(1-\pi), \mathbb{E}[J_k^{(N)*}(\pi_{next})]  \}$, it would suffice to show that $\mathbb{E}[J_k^{(N)*}(\pi_{next})]$ is concave in $\pi$.

Let us define a function:
\begin{equation*}
    \Phi_{k}(\pi,\tilde{\bm{z}}_{1:k})=J_k^{(N)*}\left(\frac{p_c(\tilde{\bm{z}}_k|t\leq k)(\pi(1-\theta)+\theta)}{p_c(\tilde{\bm{z}}_k)}\right)p_c(\tilde{\bm{z}}_k)
\end{equation*}
where $\tilde{\bm{z}}_{1:k}$ represents $\tilde{\bm{z}}_1,\tilde{\bm{z}}_2,\ldots,\tilde{\bm{z}}_{k}$. Since $\int\ldots d(\tilde{\bm{z}}_k)$ is a linear operator and $\mathbb{E}[J_k^{(N)*}(\pi_{next})]=\int\Phi_{k}(\pi,\tilde{\bm{z}}_{1:k})d(\tilde{\bm{z}}_k)$, it is sufficient to show that $\Phi_{k}(\pi,\tilde{\bm{z}}_{1:k})$ is concave in $\pi$. Since $J_k^{(N)*}(\pi)$ is assumed to be concave in $\pi$, we can write:
\begin{equation}
   J_k^{(N)*}(\pi)=\inf_{(x,y)\in V}{(x\pi+y)}
\end{equation}
where $V=\{(x,y)\in \mathbb{R}^2: x\pi+y \geq J_k^{(N)*}(\pi)\text{ }\forall\text{ } \pi\in[0,1]\}$. Hence,  

\footnotesize
\begin{eqnarray}
&& \Phi_{k}(\pi,\tilde{\bm{z}}_{1:k})\nonumber\\
    &=& J^{(N)*}_{k}\left(\frac{p_c(\tilde{\bm{z}}_k|t\leq k)(\pi(1-\theta)+\theta)}{p_c(\tilde{\bm{z}}_k)}\right)p_c(\tilde{\bm{z}}_k)\nonumber\\
    &=&\inf_{(x,y)\in V}\left(\frac{p_c(\tilde{\bm{z}}_k|t\leq k)(\pi(1-\theta)+\theta)}{p_c(\tilde{\bm{z}}_k)}x+y\right)p_c(\tilde{\bm{z}}_k) \nonumber\\
    &=&\inf_{(x,y)\in V}\{[p_c(\tilde{\bm{z}}_k|t\leq k)(1-\theta)x]\pi+x\theta p_c(\tilde{\bm{z}}_k|t\leq k)+yp_c(\tilde{\bm{z}}_k)\} \nonumber
\end{eqnarray}
\normalsize

Since infimum of linear functions is a concave function,  $\Phi_{k}(\pi,\tilde{\bm{z}}_{1:k})$ and hence $\mathbb{E}[\Phi_{k}(\pi,\tilde{\bm{z}}_{1:k})]$ are concave in $\pi$. Hence, $J_{k-1}^{(N)*}(\pi)$, which is minimum of two concave functions, is concave in $\pi$. Since, we have already proved that $J_N^{(N)*}(\pi)$ and $J_{N-1}^{(N)*}(\pi)$ is concave in $\pi$, by backward induction, we can claim that $J^{(N)*}(\pi)$ is concave in $\pi$.

\section{Proof of Theorem~\ref{theorem:threshold-policy-structure}}
\label{appendix:proof-of-threshold-policy-structure}
Let us define:

\footnotesize
\begin{eqnarray}
        g(\pi_k)&=&\lambda(1-\pi_k) \nonumber \\
      h_k(\pi_k, \tilde{\bm{z}}_1, \tilde{\bm{z}}_2, \cdots, \tilde{\bm{z}}_k) 
       &=&\pi_k+\mathbb{E}[J^*_{k+1}(\Psi_k(\pi_k, \tilde{\bm{z}}_1, \tilde{\bm{z}}_2, \cdots, \tilde{\bm{z}}_{k}))] \nonumber
\end{eqnarray}
\normalsize

Clearly $g(0)=\lambda$ and $g(1)=0$. Now note that,
\begin{eqnarray}
    && h_{k}(0, \tilde{\bm{z}}_1, \tilde{\bm{z}}_2, \cdots, \tilde{\bm{z}}_{k}) \nonumber\\
    &=&\mathbb{E}_{\tilde{\bm{z}}_{k+1}|\tilde{\bm{z}}_1\ldots\tilde{\bm{z}}_{k}}\left[J^*_{k+1}\left(\frac{\theta p_c(\tilde{\bm{z}}_{k+1}|t\leq k+1)}{p(\tilde{\bm{z}}_{k+1}|\tilde{\bm{z}}_{k}\ldots \tilde{\bm{z}}_{1})}\right)\right]\nonumber \\
    &\leq & J^*_{k+1}\left(\theta\times\mathbb{E}_{\tilde{\bm{z}}_{k+1}|\tilde{\bm{z}}_1\ldots\tilde{\bm{z}}_{k}}\left[\frac{ p_c(\tilde{\bm{z}}_{k+1}|t\leq k+1)}{p(\tilde{\bm{z}}_{k+1}|\tilde{\bm{z}}_{k}\ldots \tilde{\bm{z}}_{1})}\right]\right)\nonumber\\
    &=& J_{k+1}^*(\theta)\nonumber\\
    &\leq& \lambda(1-\theta)\nonumber\\
    &<& \lambda  \nonumber
\end{eqnarray}
where the first inequality follows from Jensen's inequality. Also, note that, $h_{k}(1, \tilde{\bm{z}}_1, \tilde{\bm{z}}_2, \cdots, \tilde{\bm{z}}_{k})=1+\mathbb{E}[J^*_{k+1}(\pi_{k+1})]\geq 1$ for any $k  \geq 1$. From the fact that $h_{k}(1, \tilde{\bm{z}}_1, \tilde{\bm{z}}_2, \cdots, \tilde{\bm{z}}_{k})-g(1)>0$ and $h_{k}(0, \tilde{\bm{z}}_1, \tilde{\bm{z}}_2, \cdots, \tilde{\bm{z}}_{k})-g(0)<0$ for all $k \geq 1$, by using the intermediate value theorem for continuous functions, we get that there exists $\Gamma_{k}(\tilde{\bm{z}}_{1:k})  \in (0,1)$ such that $h_k(\Gamma_{k}(\tilde{\bm{z}}_{1:k}), \tilde{\bm{z}}_{1:k} )=g(\Gamma_{k}(\tilde{\bm{z}}_{1:k}) )$. Further, since $h_k(\pi, \tilde{\bm{z}}_{1:k})-g(\pi)$ is a concave function, $h_k(1, \tilde{\bm{z}}_{1:k})-g(1)>0$ and $h_k(0, \tilde{\bm{z}}_{1:k})-g(0)<0$, the value of $\Gamma_{k}(\tilde{\bm{z}}_{1:k}) $ is unique for each $k \geq 1$. Hence the optimal stopping time $\tau^*$ is given by
\begin{equation}
    \tau^*=\inf\{k \geq 1:\pi_k \geq \Gamma_k (\tilde{\bm{z}}_{1:k}) \} \nonumber
\end{equation}
where $\Gamma_{k}(\tilde{\bm{z}}_{1:k}) $ is given by
\begin{eqnarray}
   && \Gamma_k (\tilde{\bm{z}}_{1:k}) +\mathbb{E}[J^*_{k+1}(\Psi_{k+1}(\Gamma_k (\tilde{\bm{z}}_{1:k}) , \tilde{\bm{z}}_1\ldots\tilde{\bm{z}}_k))] \nonumber\\
   &=& \lambda(1-\Gamma_k (\tilde{\bm{z}}_{1:k}) ) \nonumber
\end{eqnarray}

{\small
\bibliographystyle{unsrt}
\bibliography{arpan-techreport}

\begin{thebibliography}{10}

\bibitem{gupta2020quickest}
Akanshu Gupta, Abhinava Sikdar, and Arpan Chattopadhyay.
\newblock Quickest detection of false data injection attack in remote state
  estimation.
\newblock In {\em 2021 IEEE International Symposium on Information Theory
  (ISIT)}, pages 3068--3073, 2021.

\bibitem{chattopadhyay2020dynamic}
Arpan Chattopadhyay and Urbashi Mitra.
\newblock Dynamic sensor subset selection for centralized tracking of an iid
  process.
\newblock {\em IEEE Transactions on Signal Processing}, 2020.

\bibitem{guan2018distributed}
Yanpeng Guan and Xiaohua Ge.
\newblock Distributed attack detection and secure estimation of networked
  cyber-physical systems against false data injection attacks and jamming
  attacks.
\newblock {\em IEEE Transactions on Signal and Information Processing over
  Networks}, 4(1):48--59, 2018.

\bibitem{mo2009secure}
Yilin Mo and Bruno Sinopoli.
\newblock Secure control against replay attacks.
\newblock In {\em Communication, Control, and Computing, 2009. Allerton 2009.
  47th Annual Allerton Conference on}, pages 911--918. IEEE, 2009.

\bibitem{mo2014detecting}
Yilin Mo, Rohan Chabukswar, and Bruno Sinopoli.
\newblock Detecting integrity attacks on scada systems.
\newblock {\em IEEE Transactions on Control Systems Technology},
  22(4):1396--1407, 2014.

\bibitem{ding2020secure}
Derui Ding, Qing-Long Han, Xiaohua Ge, and Jun Wang.
\newblock Secure state estimation and control of cyber-physical systems: A
  survey.
\newblock {\em IEEE Transactions on Systems, Man, and Cybernetics: Systems},
  51(1):176--190, 2020.

\bibitem{guo2017optimal}
Ziyang Guo, Dawei Shi, Karl~Henrik Johansson, and Ling Shi.
\newblock Optimal linear cyber-attack on remote state estimation.
\newblock {\em IEEE Transactions on Control of Network Systems}, 4(1):4--13,
  2017.

\bibitem{chen2017optimal}
Yuan Chen, Soummya Kar, and Jos{\'e}~MF Moura.
\newblock Optimal attack strategies subject to detection constraints against
  cyber-physical systems.
\newblock {\em IEEE Transactions on Control of Network Systems}, 2017.

\bibitem{chen2016cyber}
Yuan Chen, Soummya Kar, and Jos{\'e}~MF Moura.
\newblock Cyber physical attacks with control objectives and detection
  constraints.
\newblock In {\em Decision and Control (CDC), 2016 IEEE 55th Conference on},
  pages 1125--1130. IEEE, 2016.

\bibitem{pasqualetti2013attack}
Fabio Pasqualetti, Florian D{\"o}rfler, and Francesco Bullo.
\newblock Attack detection and identification in cyber-physical systems.
\newblock {\em IEEE Transactions on Automatic Control}, 58(11):2715--2729,
  2013.

\bibitem{li2017detection}
Yuzhe Li, Ling Shi, and Tongwen Chen.
\newblock Detection against linear deception attacks on multi-sensor remote
  state estimation.
\newblock {\em IEEE Transactions on Control of Network Systems}, 2017.

\bibitem{miao2017coding}
Fei Miao, Quanyan Zhu, Miroslav Pajic, and George~J Pappas.
\newblock Coding schemes for securing cyber-physical systems against stealthy
  data injection attacks.
\newblock {\em IEEE Transactions on Control of Network Systems}, 4(1):106--117,
  2017.

\bibitem{guo2018secure}
Ziyang Guo, Dawei Shi, Daniel~E Quevedo, and Ling Shi.
\newblock Secure state estimation against integrity attacks: A gaussian mixture
  model approach.
\newblock {\em IEEE Transactions on Signal Processing}, 67(1):194--207, 2018.

\bibitem{mishra2017secure}
Shaunak Mishra, Yasser Shoukry, Nikhil Karamchandani, Suhas~N Diggavi, and
  Paulo Tabuada.
\newblock Secure state estimation against sensor attacks in the presence of
  noise.
\newblock {\em IEEE Transactions on Control of Network Systems}, 4(1):49--59,
  2017.

\bibitem{pajic2017attack}
Miroslav Pajic, Insup Lee, and George~J Pappas.
\newblock Attack-resilient state estimation for noisy dynamical systems.
\newblock {\em IEEE Transactions on Control of Network Systems}, 4(1):82--92,
  2017.

\bibitem{liu2017dynamic}
Chensheng Liu, Jing Wu, Chengnian Long, and Yebin Wang.
\newblock Dynamic state recovery for cyber-physical systems under switching
  location attacks.
\newblock {\em IEEE Transactions on Control of Network Systems}, 4(1):14--22,
  2017.

\bibitem{nakahira2018attack}
Yorie Nakahira and Yilin Mo.
\newblock Attack-resilient h2, h-infinity, and l1 state estimator.
\newblock {\em IEEE Transactions on Automatic Control}, 2018.

\bibitem{chattopadhyay2019security}
Arpan Chattopadhyay and Urbashi Mitra.
\newblock Security against false data injection attack in cyber-physical
  systems.
\newblock {\em IEEE Transactions on Control of Network Systems}, 2019.

\bibitem{chattopadhyay2018secure}
Arpan Chattopadhyay, Urbashi Mitra, and Erik~G Str{\"o}m.
\newblock Secure estimation in v2x networks with injection and packet drop
  attacks.
\newblock In {\em 2018 15th International Symposium on Wireless Communication
  Systems (ISWCS)}, pages 1--6. IEEE, 2018.

\bibitem{chattopadhyay2018attack}
Arpan Chattopadhyay and Urbashi Mitra.
\newblock Attack detection and secure estimation under false data injection
  attack in cyber-physical systems.
\newblock In {\em 2018 52nd Annual Conference on Information Sciences and
  Systems (CISS)}, pages 1--6. IEEE, 2018.

\bibitem{manandhar2014detection}
Kebina Manandhar, Xiaojun Cao, Fei Hu, and Yao Liu.
\newblock Detection of faults and attacks including false data injection attack
  in smart grid using kalman filter.
\newblock {\em IEEE transactions on control of network systems}, 1(4):370--379,
  2014.

\bibitem{liang2017review}
Gaoqi Liang, Junhua Zhao, Fengji Luo, Steven~R Weller, and Zhao~Yang Dong.
\newblock A review of false data injection attacks against modern power
  systems.
\newblock {\em IEEE Transactions on Smart Grid}, 8(4):1630--1638, 2017.

\bibitem{hu2017secure}
Qie Hu, Dariush Fooladivanda, Young~Hwan Chang, and Claire~J Tomlin.
\newblock Secure state estimation and control for cyber security of the
  nonlinear power systems.
\newblock {\em IEEE Transactions on Control of Network Systems}, 2017.

\bibitem{fawzi2014secure}
Hamza Fawzi, Paulo Tabuada, and Suhas Diggavi.
\newblock Secure estimation and control for cyber-physical systems under
  adversarial attacks.
\newblock {\em IEEE Transactions on Automatic Control}, 59(6):1454--1467, 2014.

\bibitem{guan2017distributed}
Yanpeng Guan and Xiaohua Ge.
\newblock Distributed attack detection and secure estimation of networked
  cyber-physical systems against false data injection attacks and jamming
  attacks.
\newblock {\em IEEE Transactions on Signal and Information Processing over
  Networks}, 4(1):48--59, 2017.

\bibitem{satchidanandan2016dynamic}
Bharadwaj Satchidanandan and Panganamala~R Kumar.
\newblock Dynamic watermarking: Active defense of networked cyber--physical
  systems.
\newblock {\em Proceedings of the IEEE}, 105(2):219--240, 2016.

\bibitem{ge2019distributed}
Xiaohua Ge, Qing-Long Han, Maiying Zhong, and Xian-Ming Zhang.
\newblock Distributed krein space-based attack detection over sensor networks
  under deception attacks.
\newblock {\em Automatica}, 109:108557, 2019.

\bibitem{dorfler2011distributed}
Florian D{\"o}rfler, Fabio Pasqualetti, and Francesco Bullo.
\newblock Distributed detection of cyber-physical attacks in power networks: A
  waveform relaxation approach.
\newblock In {\em 2011 49th Annual Allerton Conference on Communication,
  Control, and Computing (Allerton)}, pages 1486--1491. IEEE, 2011.

\bibitem{choraria2019optimal}
Moulik Choraria, Arpan Chattopadhyay, Urbashi Mitra, and Erik Strom.
\newblock Optimal deception attack on networked vehicular cyber physical
  systems.
\newblock In {\em 2019 53rd Asilomar Conference on Signals, Systems, and
  Computers}, pages 1131--1135. IEEE, 2019.

\bibitem{moradi2019coordinated}
Ashkan Moradi, Naveen~KD Venkategowda, and Stefan Werner.
\newblock Coordinated data-falsification attacks in consensus-based distributed
  kalman filtering.
\newblock In {\em 2019 IEEE 8th International Workshop on Computational
  Advances in Multi-Sensor Adaptive Processing (CAMSAP)}, pages 495--499. IEEE,
  2019.

\bibitem{lu2019malicious}
An-Yang Lu and Guang-Hong Yang.
\newblock Malicious attacks on state estimation against distributed control
  systems.
\newblock {\em IEEE Transactions on Automatic Control}, 65(9):3911--3918, 2019.

\bibitem{poor2009quickest}
H~Vincent Poor and Olympia Hadjiliadis.
\newblock {\em Quickest detection}, volume~40.
\newblock Cambridge University Press Cambridge, 2009.

\bibitem{spall92original-SPSA}
J.C. Spall.
\newblock Multivariate stochastic approximation using a simultaneous
  perturbation gradient approximation.
\newblock {\em IEEE Transactions on Automatic Control}, 37(3):332--341, 1992.

\bibitem{borkar08stochastic-approximation-book}
Vivek~S. Borkar.
\newblock {\em Stochastic approximation: a dynamical systems viewpoint.}
\newblock Cambridge University Press, 2008.

\bibitem{anderson1979optimal}
Brian~DO Anderson and John~B Moore.
\newblock Optimal filtering.
\newblock {\em Englewood Cliffs}, 21:22--95, 1979.

\bibitem{poor2008quickest}
H~Vincent Poor and Olympia Hadjiliadis.
\newblock {\em Quickest detection}.
\newblock Cambridge University Press, 2008.

\bibitem{bertsekas07dynamic-programming-optimal-control-1}
D.P. Bertsekas.
\newblock {\em Dynamic Programming and Optimal Control, Vol. I}.
\newblock Athena Scientific, 2007.

\bibitem{l1983theory}
L.E. L.
\newblock {\em Theory of Point Estimation}.
\newblock Wiley series in probability and mathematical statistics. John Wiley
  and sons, 1983.

\bibitem{premkumar2008Sleepwake}
K.~Premkumar and A.~Kumar.
\newblock Optimal sleep-wake scheduling for quickest intrusion detection using
  wireless sensor networks.
\newblock In {\em IEEE INFOCOM 2008 - The 27th Conference on Computer
  Communications}, pages 1400--1408, 2008.

\bibitem{737522}
Tze~Leung Lai.
\newblock Information bounds and quick detection of parameter changes in
  stochastic systems.
\newblock {\em IEEE Transactions on Information Theory}, 44(7):2917--2929,
  1998.

\bibitem{10.1093/biomet/41.1-2.100}
E.~S. PAGE.
\newblock {CONTINUOUS INSPECTION SCHEMES}.
\newblock {\em Biometrika}, 41(1-2):100--115, 06 1954.

\bibitem{10.1214/aoms/1177693055}
G.~Lorden.
\newblock {Procedures for Reacting to a Change in Distribution}.
\newblock {\em The Annals of Mathematical Statistics}, 42(6):1897 -- 1908,
  1971.

\bibitem{10.1214/aos/1176350164}
George~V. Moustakides.
\newblock {Optimal Stopping Times for Detecting Changes in Distributions}.
\newblock {\em The Annals of Statistics}, 14(4):1379 -- 1387, 1986.

\bibitem{10.1214/aos/1176347761}
Y.~Ritov.
\newblock {Decision Theoretic Optimality of the Cusum Procedure}.
\newblock {\em The Annals of Statistics}, 18(3):1464 -- 1469, 1990.

\bibitem{10.1214/aos/1176346587}
Moshe Pollak.
\newblock {Optimal Detection of a Change in Distribution}.
\newblock {\em The Annals of Statistics}, 13(1):206 -- 227, 1985.

\bibitem{veeravalli2012quickest}
Venugopal~V. Veeravalli and Taposh Banerjee.
\newblock Quickest change detection, 2012.

\end{thebibliography}
}

\end{document}